\documentclass[11pt, draftcls, onecolumn]{IEEEtran}

\usepackage{color,epsf,psfrag,amssymb,amsfonts,latexsym,cite,verbatim}
\usepackage{enumerate,times,array, float}
\usepackage{amsmath,cases,graphicx}
\usepackage[mathscr]{eucal}

\def\psfancypar#1#2{\begingroup\def\par{\endgraf\endgroup\lineskiplimit=0pt}
               \setbox2=\hbox{\large\sc #2}
               \newdimen\tmpht \tmpht \ht2 \advance\tmpht by \baselineskip
               \font\hhuge=Times-Bold at \tmpht
               \setbox1=\hbox{{\hhuge #1}}
               \count7=\tmpht \count8=\ht1
               \divide\count8 by 1000 \divide\count7 by \count8 
               \tmpht=.001\tmpht\multiply\tmpht by \count7 
               \font\hhuge=Times-Bold at \tmpht
               \setbox1=\hbox{{\hhuge #1}}
               \noindent
                \hangindent1.05\wd1
               \hangafter=-2 {\hskip-\hangindent
               \lower1\ht1\hbox{\raise1.0\ht2\copy1}%
                \kern-0\wd1}\copy2\lineskiplimit=-1000pt}

\newcommand{\Phibf}{\mbox{${\bf \Phi}$}}

\newcommand{\E}{\mbox{{\rm E}}}

\newcommand{\abf}{\mbox{${\bf a}$}}

\def\boxit#1{\vbox{\hrule\hbox{\vrule\kern3pt
        \vbox{\kern3pt#1\kern3pt}\kern3pt\vrule}\hrule}}

\def\reals{ { {\rm  I \kern-0.15em R }  } }
\def\complex{ {\,{{\rm C} \kern-0.50em \raise0.20ex {  |}}\, }}

\def\lambdabf{\hbox{\boldmath$\lambda$\unboldmath}}
\def\mubf{\hbox{\boldmath$\mu$\unboldmath}}

\def\Sigmabf{\hbox{$\bf \Sigma$}}

\def\Pibf{{\bf \Pi}}
\def\abf{{\bf a}}

\def\hbf{{\bf h}}

\def\nbf{{\bf n}}

\def\pbf{{\bf p}}

\def\sbf{{\bf s}}

\def\ubf{{\bf u}}
\def\vbf{{\bf v}}
\def\wbf{{\bf w}}
\def\xbf{{\bf x}}
\def\ybf{{\bf y}}
\def\zbf{{\bf z}}

\def\xbf{{\bf x}}
\def\ybf{{\bf y}}
\def\Abf{{\bf A}}
\def\Bbf{{\bf B}}

\def\Hbf{{\bf H}}
\def\Ibf{{\bf I}}

\def\Qbf{{\bf Q}}
\def\Rbf{{\bf R}}

\def\Ubf{{\bf U}}
\def\Vbf{{\bf V}}

\def\Xbf{{\bf X}}

\def\Cc{{\cal C}}

\def\Fc{{\cal F}}

\def\Hc{{\cal H}}

\def\Rc{{\cal R}}

\def\be{\vskip .3cm \begin{equation}}
\def\ee{\end{equation} \vskip .4cm \noindent}
\newcommand{\R}{\mbox{$\hat {\bf R}_{N}$}}

\def\Rxx{\Rbf_{\ssstyle X\kern-.1em X}}

\let\ssstyle=\scriptscriptstyle

\def\Kout{\setbox1=\hbox{\Huge\bf K}\hbox to
1.05\wd1{\hspace{.05\wd1}
}}
\def\Sout{\setbox1=\hbox{\Huge\bf S}\hbox to 1.05\wd1{\hspace{.05\wd1}
}}

  \ifx\LabelFigloaded\MYundefined\relax
  \else
   \fi

  \def\LabelFigloaded{\relax}% now loaded

  \chardef\LabelFigCatAt\the\catcode`\@
  \catcode`\@=11

 \let\LabelFigwlog@ld\wlog
 \def\wlog#1{\relax}

 \ifx\\\MYundefined@
    \let\\\relax
 \fi

  \def\ms@g{\immediate\write16}

 \def\N@wif{\csname newif\endcsname }
 \def\Temp@ {\N@wif\ifIN@}
 \ifx\INN@\MYundefined@
    \else \let\Temp@\relax
 \fi
 \Temp@

  \def\IN@{\expandafter\INN@\expandafter}
  \long\def\INN@0#1@#2@{\long\def\NI@##1#1##2##3\ENDNI@
    {\ifx\m@rker##2\IN@false\else\IN@true\fi}%
     \expandafter\NI@#2@@#1\m@rker\ENDNI@}
  \def\m@rker{\m@@rker}
 
  \newtoks\Initialtoks@  \newtoks\Terminaltoks@
  \def\SPLIT@{\expandafter\SPLITT@\expandafter}
  \def\SPLITT@0#1@#2@{\def\TTILPS@##1#1##2@{%
     \Initialtoks@{##1}\Terminaltoks@{##2}}\expandafter\TTILPS@#2@}

 \def\Shifted@@#1#2#3{\setbox0=\hbox{#3}%
   \raise -\dp0\vbox {\kern-#2%
       \hbox {\kern#1\unhbox0\kern-#1}%
           \kern#2}}

 \newcount\gridcount
 \newbox\auxGridbox@ \newbox\hGridbox@ \newbox\vGridbox@
 \newbox\Labelbox@ \newbox\auxLabelbox@
 \newbox\Coordinatebox@
 \newtoks\Labeltoks@
 \newdimen\Wdd@ \newdimen\Htt@
 \newdimen\Wddd@ \newdimen\Httt@
 
 \def\Wr@{\immediate\write16}

 \newdimen\GL@wd%% grid-line width
 \def\GridLineWidth#1{\GL@wd=#1}

 \def\gobble#1{}
 \def\EdgeErr@{\Wr@{}%
      \Wr@{\string\Edges\space argument
      1, 10, 100 or 1000 please\string!}%
      }

 \newcount\Edgect@

 \def\Sweepup#1\endSweepup{}

 \def\SetEdges@{%
    \edef\Zr@@s{\expandafter\gobble\number\Edgect@\empty}%
        \count255=0\Zr@@s\relax
        \ifnum\count255=\z@\else\EdgeErr@\show\tailtest\fi
        \count255=1\Zr@@s\relax%\showthe\count255
        \ifnum\count255=\Edgect@\relax\else\EdgeErr@\show\leadtest\fi
    \EdgGl@b\edef\Zr@s{\expandafter\gobble\Zr@@s\empty}%\show\Zr@s
    \ifnum\Edgect@>\@ne\relax\EdgGl@b\let\L@Dc\empty
        \else\EdgGl@b\edef\L@Dc{\string.}\fi
    \ifnum\Edgect@>\@ne\relax
        \EdgGl@b\edef\Edgescale@##1{\divide##1 by \Edgect@}%
        \else\EdgGl@b\edef\Edgescale@##1{}\fi
    }

 \def\Edges#1{\Edgect@=#1\relax
     \let\EdgGl@b\global \SetEdges@}

 \Edges{1}%% default

 \def\hhrule{\hrule height \GL@wd\vskip-.\GL@wd}

 \def\hRule@{%
   \advance\gridcount -2%
   \vfil\hhrule\vfil
   \llap{\smash{\raise -2.5pt
     \hbox{\L@Dc\number\gridcount\Zr@s\kern2pt}}}%
   \hhrule
   }

\def\vvrule{\vrule width \GL@wd \kern-\GL@wd}

 \def\vRule@{\advance\gridcount 2%
   \hfil\vvrule\hfil
   \setbox\auxGridbox@=\vbox to 0pt
      {\vskip \Htt@\vskip 2pt
        \hbox to 0pt{\hss\L@Dc\number\gridcount\Zr@s\hss}\vss}%
      \wd\auxGridbox@=0pt \box\auxGridbox@
   \vvrule
   }

 \def\PlaceGrid@@{\gridcount=10 
  \setbox\hGridbox@=\hbox{%
        \hbox{%
             \hskip-.4pt\vrule
             \vbox to \Htt@{%
               \offinterlineskip\parindent=\z@\relax
               \hbox to \Wdd@{\hfil}
               \hRule@\hRule@\hRule@\hRule@
               \vfil\hhrule\vfil}%
             \vrule\hskip-.4pt}
    }%
  \gridcount=0%
  \setbox\vGridbox@=\hbox{%
      \vbox{\offinterlineskip\parindent=0pt\hsize=0pt
         \vskip-.4pt\hrule%
         \hbox to \Wdd@{%
                 \vtop to \Htt@{\vfil}%
                 \vRule@\vRule@\vRule@\vRule@
                 \hfil\vvrule\hfil}%
         \hrule\vskip-.4pt}}%
  \wd\hGridbox@=0pt\ht\hGridbox@=0pt
  \wd\vGridbox@=0pt\ht\vGridbox@=0pt
  \hbox{\box\hGridbox@\box\vGridbox@}%
  }

 \def\LabelsGlobal{\def\LabGl@b{\global}}
 \def\LabelsLocal{\def\LabGl@b{}}
 \LabelsGlobal %% default

 \def\SetLabels#1\endSetLabels{%
   \LabGl@b\Labeltoks@={#1()\\}%
   }

 \LabGl@b\Labeltoks@={()\\}

 \def\ShowGrid{\LabGl@b\let\PlaceGrid@\PlaceGrid@@}
 \def\HideGrid{\LabGl@b\let\PlaceGrid@\relax}
 \def\Grids{\ShowGrid\LabGl@b\let\GridSwitch@\ShowGrid}
 \def\noGrids{\HideGrid\LabGl@b\let\GridSwitch@\HideGrid}

 \noGrids

 \def\bAdjust@@{%
     \setbox\auxLabelbox@=\hbox{\raise \dp\auxLabelbox@
            \box\auxLabelbox@}}
 \def\bAdjust@{\let\vAdjust@\bAdjust@@}

 \def\eAdjust@@{\dimen0=-.5\ht\auxLabelbox@
     \advance\dimen0 by .5\dp\auxLabelbox@
     \setbox\auxLabelbox@=
            \hbox{\raise\dimen0\box\auxLabelbox@}}
 \def\eAdjust@{\let\vAdjust@\eAdjust@@}

 \def\tAdjust@@{%
     \setbox\auxLabelbox@=\hbox{\raise-\ht\auxLabelbox@
            \box\auxLabelbox@}}
 \def\tAdjust@{\let\vAdjust@\tAdjust@@}

 \let\vAdjust@\relax

 \def\lAdjust@{\let\hAdjust@\rlap}
 \def\rAdjust@{\let\hAdjust@\llap}

 \let\hAdjust@\relax\let\vAdjust@\relax

 \def\FetchLabel@#1(#2)#3\\{%
     \IN@0#2@@\ifIN@
        \setbox0=\hbox{\ignorespaces#1#3\unskip}%
        \ifdim\wd0>0pt
           \ms@g{}%
           \ms@g{ !!! Bad label(s)? !!!}%
           \message{ #1(#2)#3}%
        \fi
        \def\LabelMole@##1\endFetchLabel@{%
            \IN@0()\\@##1@%
            \ifIN@\def\Temp@{\FetchLabel@##1\endFetchLabel@}%
            \else\def\Temp@{}%
            \fi
            \Temp@
           }%
     \else
       \ignorespaces#1\unskip
       \setbox\auxLabelbox@=%
         \hbox to 0pt{\hss\ignorespaces\hAdjust@
          {\ignorespaces#3\unskip}\hss}%
       \vAdjust@
       \let\hAdjust@\relax\let\vAdjust@\relax
       \AugmentLabelBox@@{#2}%
       \ht\Labelbox@=0pt\dp\Labelbox@=0pt
       \let\LabelMole@\FetchLabel@%
     \fi\LabelMole@}

 \newtoks\XYSep@ %\XYSep@{*}
 \def\SetXYSeparator#1{%
     \IN@0#1@@\ifIN@\XYSep@{*}%
     \else
     \XYSep@{#1}%
     \fi
     }

 \SetXYSeparator*

 \def\AugmentLabelBox@@#1{%
     \IN@0\the\XYSep@ @#1@\ifIN@
       \SPLIT@0\the\XYSep@ @#1@%
       \setbox\Labelbox@=\hbox to 0pt{%
         \unhbox\Labelbox@
         \Shifted@@{\the\Initialtoks@\Wddd@}%
         {\the\Terminaltoks@\Httt@}%
         {\box\auxLabelbox@}}%
     \else
         \ms@g{}%
         \ms@g{ !!! Bad insertion point. !!!}%
         \message{ (#1\ this point was rejected.)}%
     \fi
    }

 \def\FetchOption@#1[#2]#3\endFetchOption@{%
    \def\temp{#1}%\show\temp
    \ifx\temp\empty
       \Edgect@=#2\relax%\showthe\Edgect@
       \let\EdgGl@b\relax
       \SetEdges@%\def\Edgescale@##1{\divide##1 by \Edgect@\relax}%
       \Cleaner@#3%
    \fi}

 \def\Cleaner@#1[@]{\Labeltoks@{#1}}
     
 \def\PlaceLabels@@{\mathsurround=0pt%\bgroup
     \def\Cr@{\\}%
     \let\L\lAdjust@\let\R\rAdjust@
     \let\B\bAdjust@\let\E\eAdjust@\let\T\tAdjust@
     \expandafter\FetchOption@\the\Labeltoks@[@]\endFetchOption@
     \Wddd@=\Wdd@ \Edgescale@\Wddd@ %\showthe\Edgect@
     \Httt@=\Htt@ \Edgescale@\Httt@
     \expandafter\FetchLabel@\the\Labeltoks@\endFetchLabel@
     \box\Labelbox@%\egroup
     }%

 \let \PlaceLabels@\PlaceLabels@@

 \def\AffixLabels#1{\setbox\Coordinatebox@=\hbox{#1}%
      \Wdd@=\wd\Coordinatebox@ \Htt@=\ht\Coordinatebox@
      \advance\Htt@ \dp\Coordinatebox@
      \hbox{\copy\Coordinatebox@\kern-\Wdd@ 
           \Shifted@@{0pt}{-\dp\Coordinatebox@}%
           {\PlaceLabels@\PlaceGrid@}%
           \kern\Wdd@}%
      \GridSwitch@ %% next grid hidden
      \LabGl@b\Labeltoks@{()\\}%
      }
 
   \let\wlog\LabelFigwlog@ld   %%restore logging
   \catcode`\@=\LabelFigCatAt  %%12 or 13

                                By

              Raymond S\'eroul <A18645@FRCCSC21.BITNET>
                                and 
              Laurent Siebenmann <lcs@topo.math.u-psud.fr>
    
              VERSIONS: July 1991, Oct 1991, Jan 1992, July 1992

INTRODUCTION

      This labelling package is intended for TeX users who
rely on non-TeX sources for for their graphics inserts.  It
provides means for adding TeX labels to such inserts with a
minimum of fuss. 

       For most labels, TeX users have in the past found it
reasonably convenient to rely on non-TeX sources. Typical
occasions when an inescapable need for TeX labels seemed to
arise are

 (a) when the graphics program lacks certain exotic or complex
mathematical symbols

 (b) when the very highest typographical quality is wanted for the
labels

 (c) when labels included with the graphics fail to print, 
 and you cannot figure out why (cf. boxedeps.doc).  The labels
 provided by labelfig.tex are 100% portable.

       Since this package first appeared, many users, who in the
past scarcely dreamed of using TeX labels, have come to use
nothing but.  So it is now appropriate to add

Intoxication Warning:  TeX labels may be addictive and expensive. 

     If you have a fast preview you may disagree, and even find
that this package provides an agreeable paste-up environment; see
extra applications at end.

     Note to publishers: It is possible and convenient to ultimately
export the TeX labels produced by labelfig.tex to become an integral
part of the EPS file. This is often desired by a publisher who typically
uses an "upmarket" graphics or page layout program, with which the
staff is skilled in perfecting figures.  See Appendix I for
a recipe.

     The authors are grateful to Patrick Ion of Math Reviews for
helpful comments and encouragement.

BASIC INSTRUCTIONS

    After reading in the macro file using

preview or proof your figure with a coordinate grid printed on
top, by typing the following:

    \ShowGrid  % shows grid  for next figure only
    \AffixLabels{<the graphics insertion>}

Here <the graphics insertion> is what you would type to insert
the graphics object alone without the grid.  This must provide
for the space around it. For example <the graphics insertion>
might well be \BoxedEPSF{MyFigure scaled 700} using the
boxedeps.tex macro package (from same source); this provides a
TeX box containing the encapsulated PostScript insert specified by
the file MyFigure. \AffixLabels{...} provides the grid (supposing
\ShowGrid is present) and later, once you have specified labels
using the grid, it will "tack on" the labels.

     The grid is a sort of (usually elongated) checkerboard of
ten rows and ten columns and its (internal) partitions are by
default numbered  .1, ... ,.9  both horizontally (X-coordinate
running left to right) and vertically (Y-coordinate running bottom
to top).  Thus the points enclosed by the grid correspond to the
points of the unit square in the cartesian "X-Y" plane, the lower
left corner corresponding to the origin (0,0).  By extrapolation,
the full page corresponds to a larger rectangle in the plane.

     These coordinates serve to position labels as follows.
Before the \AffixLabels{...} command type label specifications:

  \SetLabels
   (<X-coordinate>*<Y-coordinate>) <first label> \\
   .
   .
   .
   (<X-coordinate>*<Y-coordinate>)  <last label> \\
  \endSetLabels

Each row specifies one label and is terminated by \\.  In each
row, the position indicator comes first; it is written as a
standard cartesian point except that the X- and Y- coordinates
are separated by * rather than a comma because TeX allows a
comma as decimal point. There are no dimension units to specify
as the unit is the grid itself.

     By default, this cartesian point specifies where the middle
of the baseline of the label will be located.  However if you precede
the point by \L [or \R] the left [or right] edge of the baseline will
be located there. Similarly you may also precede the point by \T, \E,
or \B to vertically align the top equator or bottom of the label box
at the specified point.  This gives nine standard positions of
the label with respect to the insertion point --- corresponding to
the eight principle points of the compas and the center

                     \L\T     \T      \R\T

                     \L\E     \E      \R\E

                     \L\B     \B      \R\B

But this neglects the default "baseline" level of TeX,
giving potentially three more positions

                     \L    <no tag>   \R

For text, the baseline level is often the preferred. Its relation to
the others is variable. It will often coincide with the bottom level,
as happens for "X".  But it is often distinct, as for "g", in which
case you have in all 12 distinct positions rather than 9.

     It is convenient to think of this specification of label
position as attaching the label by a thumb-tack to the coordinate
grid. There are up to twelve positions of the thumb-tack on the
label, while the position of the thumb-tack on the coordinate grid is
arbitrary.  Normally, one choses the position of the thumb-tack on
the label to be the one that is the closest to the item being
labeled.  There are good reasons for this "rule of thumb":

   (a)  It facilitates correct positioning at first try.

   (b)  If the scale of the figure must be altered after labels
have been affixed, the labels have a good chance of remaining well
positioned.

   (c)  The visible grid need not extend beyond the "bounding box"
for the figure, because the best preferred position is always
(at least almost) within the bounding box .

The second reason is particularly important. Indeed it often
happens that scale has to be altered after labelling begins, in
order to either provide space for the labels, or to adjust
proportions between the labels and the figure.  (The size of labels
is unaffected by scaling.)

     Here is an artificial but self-contained test which uses
TeX rules to make a graphics object.

TEST

    Do not skip this!

 \def\FrameIt#1{\hbox{\vrule$\vcenter {\hrule\kern3pt%
             \hbox {\kern3pt #1\kern3pt}%
               \kern3pt\hrule}$\relax\vrule}}

 \def\Caption#1#2{\FrameIt{%
       \vtop {\hsize=#1\relax \parindent=0pt
         \leftskip=0pt \rightskip=0pt plus15pt
         \parfillskip=0pt
         \lineskip=1pt\baselineskip=0pt
         #2}}}

 \def\FirstQuadrant{\hbox to 100pt{\vrule\vbox to 100pt{%
        \hbox to 100pt{\hfil}\vfil\hrule}\hss}}

  \SetLabels
    \R(.5*.2) $\zeta\,\cdot$\\
    (.9*-.10) $\xi$\\
    \R(-.03*.9) $\eta$\\
    \T(.5*.9) \Caption{70pt}{%
          \it The norm of
          $g(\xi+i\eta)$ is indicated on
          contours of this invisible surface.}\\
  \endSetLabels

  \AffixLabels{\FirstQuadrant}

  \end

  Note that the coordinates to use for labels are indicated on the
edges of the grid (when visible) corresponding to the conventional
x- and y- axes of the Cartesian plane. By default the grid is
1-by-1. However, by the command \Edges{100}, you can change this
to 100-by-100 and many users find this alternative most
convenient. Place the command \Edges{...} in your style file (or
header) since its effect is is global. Other possible edge values
are 10 and 1000.

  If you use the command \Edges{...} at all, do so with care.  For
if you accidentally delete an \Edges{...} command your labels will
abruptly be badly misplaced and may logically but mysteriously
generate "dimension too big" errors under TeX and "off page" errors
under your driver.  

  You can dictate the edgescale for an individual figure by giving
the scale in brackets immediately after \SetLabels.  Thus, to
import into an article using say \Edge{100} a figure labelled using
another edgescale, say the original 1-by-1 default, you can use
\SetLabels[1]...\endSetLabels.

GETTING IT DOWN PAT

     Complicated labeling deserves the same respect as
complicated mathematics.  Do not expect it to come out perfect the
first time!  What is needed in either case is a mechanism to
repeatedly typeset troublesome pieces.

     One mechanism is always available.  One does complicated
labelling in a separate "test" file involving just the figure being
labelled;  a texpert will know how to \dump TeX's current state as
a temporary format that restarts rapidly at each retry.  Usually,
one then pastes the completed labelled figure back into the main
TeX file, but, of course, one can also \input it as an auxiliary
file.

     If you do not have a TeXpert at handy, here is a first
approximation to an efficient setup. By deletions reduce a copy
of your article to just a few lines before and after the figure.
Now label the figure, and finally, copy and paste the labelled
figure to the original article. Then copy the next figure to label
into this testbed and repeat. The TeXpert can improve the  speed
at which TeX starts up, by compiling a format specifically for
your article; just one caution: best NOT include in the format
ephemeral details of setup like \Set<mydriver>ArtSpecials (from
boxedeps.tex because this reads  figure dimensions which you may
change during your work session.

     An improved mechanism to repeatedly typeset troublesome
pieces is now available on the Macintosh; it is called LinoTeX;
see the same ftp sources.  It could be set up on many types
of computer.

     Before using labelfig.tex to attach labels to a graphics
object inserted using boxedeps.tex or BoxedArt.tex, make it a
firm rule to carefully adjust the bounding box using the trimming
commands of these packages, and also at least tentatively scale
and position the object. Beware of changing the grid inadvertently
after the labels have been positioned.  For example, correcting
the bounding box of a PostScript graphics object can foul up the
labels by changing the coordinate grid to which the labels are
attached. This is particularly true for the trimming  commands of
boxedeps.tex and BoxedArt.tex. However, as noted already, change
of scale is much less disruptive, and modest adjustments should be
well tolerated.

     Sometimes the labels protrude so far from the bounding box
of a figure that the figure has to be repositioned.  Best do this
by ad hoc spacing, say using \hglue and \vglue; altering the
bounding box would create a vicious circle.

     Remember that you are responsible for preventing labels
from overlapping. You are responsible for all label typography
including size and style. A label is really just about anything
that can be put in a TeX box. Note that spaces at the beginning
and end of labels will normally be suppressed; if you really want
them you must protect them with TeX braces.

     This package temporarily sets the \mathsurround parameter
of TeX to zero  while the labels are being affixed. This is done
because nonzero \mathsurround space would influence the position
of left and right aligned labels; then, when a texpert or printer
modifies mathsurround, diagram labeling might be disastrously
altered. There is a small price to pay involving labels that are
formatted as caption boxes including mathematics: you  may want or
need to specify an explicit mathsurround space within the caption
box; it will not influence anything outside.

     Those hostile to the use of * as separator between
the X and Y coordinates of label insertion points, are free to
impose another using \SetXYSeparator{<the new separator>}.  
Americans may prefer "," to "*" since they never use a 
comma as a decimal point; on the other hand, * may be more visible.

APPENDIX (I)  MERGING labelfig.tex LABELS INTO AN EPSF GRAPHICS OBJECT.

     As promised in the introduction, here is a recipe useful for
publishers. It works at least on Macintosh and at least for vectorized
graphics and Adobe type1 fonts.  (There is surely a similar recipe for
PCs under MSWindows.)

 (a)  Use boxedeps.tex utility to integrate the figure given by the eps
file, "x.eps" say, with a visible frame around it.  See
\ShowDisplacementBoxes command in boxedeps.tex.  To get precise results
automatically it is important to use the \Trim... commands of
boxedeps.tex making the "DisplacementBox" neatly fit the figure.

 (b)  Use the TeX printer driver and LaserWriter (versions >= 8.1.1) to
export to an EPSF the DVI page containing the integrated, labelled
figure. You now have an EPS file  "xx.eps"  that contains too much, and at
the wrong scale, and at wrong position.

 (c)  Convert the EPSF to an Adode Illustrator format EPSF using
the shareware utility called epsConvert by Sam Weiss
1993-- (currently $25).

 (d)  In Illustrator (or a compatible program), group the labels and the
"DisplacementBox"; copy them to the clipboard and paste them into "x.ps".
This step requires that all the label fonts be "visible to the Macintosh.

 (e)  Translate and scale the pasted group consisting of the labels plus
the "DisplacementBox" so as to make the "DisplacementBox" the bounding
box of (labelless) figure represented by "x.eps".  At this point the
labels will be correctly placed on the figure "x.eps".

 (f)  Ungroup and delete the "DisplacementBox".  The result is the
desired single EPS file, "x+.eps" say, It contains the original figure
plus its labels.  

     Using grouping and ungrouping appropriately in "x+.eps", a
publisher's staff can very efficiently improve label positions etc.

APPENDIX II)  SOME EXOTIC APPLICATIONS

     The grid of labelfig.tex is analogous to a light-table in
classical page makeup with wax or latex glue.  In principle, you
can use it to compose any page from its indivisible parts.  This
even has some of the artisanal charm of classical paste-up
provided you have a fast screen preview to make the process
"interactive".

     In practice labelfig.tex is a tool for nonstandard jobs.
Here are a few going beyond the labelling already discussed.

(I)  GRAPHICS INTEGRATION.

     This is accomplished by treating the imported graphics
objects as labels.  The underlying graphics object is then
typically an empty  \vbox to <dimension>{\vfill} in a TeX
\midinsert...\endinsert construction.  A label line
might be of the form

   (.1*.1) \\

The exact form of the special command varies from driver to
driver.  However, in the case of encapsulated PostScript graphics
(EPSF norm), by relying on boxedeps.tex, one can have the
following standard syntax (independant of driver  (see
boxedeps.doc for details.
  
  (.1*.1) \BoxedEPSF{MyFigure scaled <scale in mils>}\\

This may be slow since it requires TeX to read the PostScript
file to read bounding box using many complex macros.  So you
may want to try

  (.1*.1) \EPSFSpecial{MyFigure}{<scale in mils>}\\

which is fast and driver independant, but it squashes the
bounding box, normally to its lower left corner.

     Similarly for graphics of the Macintosh PICT norm ---
using BoxedArt.tex (same sources) in place of boxedeps.tex.

     This approach to integration is to be recommended when
one is assembling a composite graphics object.

 (II)  COMMUTATIVE DIAGRAM ENHANCEMENT

     Commutative diagrams or arrays of mathematical objects
connected by arrows of various sorts are common in mathematics.
The mathematical objects require the use of TeX.  Recently TeX
acquired a good collection of arrows of all slopes --- that of
LamSTeX --- plus pwerful macros to build the diagrams.

     However, even the LamSTeX collection is often
inadequate; it lacks for example double shafted arrows, dotted
arrows and curved arrows. Fortunately it is possible to produce
such arrows on an individual basis using sophisticated graphics
programs such as Illustrator and AldusFreehand (both serving
the EPSF norm) or using Metafont (with its public domain norm).
Since the creation of each new arrow is a work of love, you
probably want to limit the number of arrows by using LamSTeX
for most arrows. The 40K commutative diagram module of LamSTeX
has been adapted to work with AmSTeX and a copy may be posted
with LabelFig and related files. Unfortunately no one has yet
offered a version that works with Plain TeX or LaTeX.

       Suffice it here to say that when the exotic arrow has
been somehow imported into TeX, labelfig.tex treats it as a
label that one affixes to the commutative diagram.  Two other
steps will be treated in separate notes, namely the matter of
extracting the dimension specifications for the arrow and the
construction of the arrow --- for these steps are far from
unique and often depend intimately on your computer environment. 
Notes for the Macintosh-Textures-Illustrator combination are
found in the file ExoticArrows.doc.

 (III) NESTING 

Ingenuity pays off in exploiting labelfig.tex. One can
mix graphics and typography quite freely.  labelfig.tex is good
for freeform or overlapping arrangements, while boxedeps.tex (or
BoxedArt.tex) is best for regimented non-overlapping
arrangements --- and the two can be combined.

     The default behavior of labelfig.tex is not ideal 
for nesting objects, because to prevent trouble for beginners
the register for labels is globally cleared when \AffixLabels
concludes.  But there are switches available

      \LabelsGlobal      \LabelsLocal

which change this.  To understand this, extend the above test 
by something like:

 \LabelsLocal
 \SetLabels
    (.5*.5) AAA\\
 \endSetLabels

 {%%% Watch for influence of braces!!
 \SetLabels
    (.5*.5) ZZZ\\
 \endSetLabels
   \AffixLabels{\FirstQuadrant}
 }

   \AffixLabels{\FirstQuadrant}

     There are however potential pitfalls.  Neither
labelfig.tex nor boxedeps.tex has been tested under extreme
conditions. Problems may occur if their procedures are
indiscriminately nested. For boxedeps.tex (not labelfig.tex)
there is a precise cause for worry, namely many of its
variables are "global", which means that TeX braces will not
provide the protection one might expect.

COMMAND SUMMARY FOR labelfig.tex

  Here [...] means optional (one or zero)
       [...]* means any number of such constructs

  \SetLabels
    [[<P>](<X><Sep><Y>) <label> \\]*
  \endSetLabels
  \ShowGrid  % this makes the grid visible (once)
  \AffixLabels{<the figure>}

   --- <P> is tack position, one of eleven or empty
              order irrelevant

                   \L\T      \T      \R\T

                   \L\E      \E      \R\E

                     \L               \R

                   \L\B      \B      \R\B

   --- (<X><Sep><Y>) insertion point;
  <Sep> is separator, = * by default;
  \SetXYSeparator{<Sep>} changes it.
   <X> and <Y> are real numbers

  --- <label> a label to attach 

  --- <the figure> the figure to label 

  \GlobalLabels (default)     
  \LocalLabels  setting for nested constructs.

 \Grids makes ALL grids appear; \HideGrid then makes just next disappear.
 \noGrids returns to default.  The commands are always global.

 \GridLineWidth{<dimension>} adjusts width of grid lines. Default is very
small, to give "hairline" effect. If your grid lines are missing try
setting \GridLineWidth{1pt}.

 \Edges#1 globally changes the edge size of all grids to the numerical 
value #1, which must be 1, 10, 100, or 1000.  The default is 1.

VERSION HISTORY.
 --- Jan 1993: \Edges#1 and [??] option after \SetLabels
 --- July 1992: \Grids, \noGrids, \HideGrid;
       Gridlines become hairlines; \GridLineWidth{<dimension>}.
 --- Oct 1991, Jan 1992: \SetXYSeparator{<Sep>},  \LabelsGlobal,
       \LabelsLocal.
 --- July 1991: first release

Address for bugs and other feedback:

        Raymond S\'eroul
        IREM and Lab. de Typographie Informatise
        Univ. Rene Descartes
        Strasbourg

    Tel 33-88-41-63-45
    Email:  A18645@FRCCSC21.BITNET

        Laurent Siebenmann
        Mathematique, Bat. 425,
        Univ de Paris-Sud,
        91405-Orsay,
        France

    Tel 33-1-6941-7949; 
    Email: lcs@topo.math.u-psud.fr

\def\scalefig#1{\epsfxsize #1\textwidth}

\newtheorem{theorem}{Theorem}
\newtheorem{lemma}{Lemma}

\newtheorem{proposition}{Proposition}

\newtheorem{problem}{Problem}

\title{{\huge Coordinated Beamforming with Relaxed Zero Forcing: The Sequential Orthogonal Projection Combining Method and Rate Control}}

\author{
Juho Park, Gilwon Lee, {\em Student~Members, IEEE}, Youngchul
Sung$^*$\thanks{$^*$Corresponding author}, {\em
Senior~Member, IEEE}, and Masahiro Yukawa, {\em Member, IEEE}
\thanks{Juho Park, Gilwon Lee, and Youngchul Sung are with the Dept. of
Electrical Engineering, KAIST, Daejeon, 305-701, South Korea.
E-mail:\{jhp@, gwlee@, and ysung@ee.\}kaist.ac.kr and Masahiro
Yukawa is with the Dept. of Electrical and Electronic Engineering,
Niigata University, Niigata, Japan, E-mail:
yukawa@eng.niigata-u.ac.jp. This research was supported by the
Basic Science Research Program through the National Research
Foundation of Korea (NRF) funded by the Ministry of Education,
Science and Technology (2010-0021269).  Some part
of this paper was presented in  WCSP 2011
\cite{Lee&Park&Sung:11WCSP}.}
}

\begin{document}

% make the title area
\maketitle

\begin{center}
 EDICS: MSP-MULT
\end{center}

\begin{abstract}
In this paper, coordinated beamforming  based on relaxed zero
forcing (RZF) for $K$ transmitter-receiver pair multiple-input
single-output (MISO) and multiple-input multiple-output (MIMO)
interference channels is considered. In the RZF coordinated
beamforming, conventional zero-forcing interference leakage
constraints are relaxed so that some predetermined interference
leakage to undesired receivers is allowed in order to increase the
beam design space for larger rates than those of the zero-forcing (ZF)
scheme {or to make beam design feasible when ZF is impossible.} In the MISO case, it is shown that the {rate-maximizing} beam vector under the RZF framework for a given set of  interference leakage levels can be obtained by sequential orthogonal projection combining (SOPC). Based on this, exact and approximate closed-form solutions are provided in  two-user and three-user cases, respectively, and an efficient beam design algorithm for RZF coordinated beamforming is provided in general cases. Furthermore, the rate control
problem under the RZF framework is considered. A centralized
approach and a distributed heuristic approach are proposed to
control the position of the designed rate-tuple in the achievable rate region. Finally, the RZF framework
is extended to MIMO interference channels by deriving a new lower
bound on the rate of each user.
\end{abstract}


%%%%%%%%%%%%%%%%%%%%%%%%%%%%%%%%%%%%%%%%%%%%%%%%%%%%%%%%%%%
\section{Introduction}
%%%%%%%%%%%%%%%%%%%%%%%%%%%%%%%%%%%%%%%%%%%%%%%%%%%%%%%%%%%

In current and future cellular networks, handling interference in
the network is one of the most critical problems. Among the many
ways of handling interference, MIMO antenna techniques and base
station cooperation are considered as the key technologies to the
interference problem. Indeed, the 3GPP Long-Term
Evolution-Advanced  considers the base station cooperation and
MIMO techniques to mitigate inter-cell interference under the name
of Coordinated Multipoint (CoMP) \cite{3GPP:10,
Sawahashietal:10WCM}. Mathematically, when each mobile station has
a single receive antenna and data is not shared among base
stations, the system is modelled as a MISO interference channel
(IC), and extensive research has been conducted on beam design for
this MISO IC, especially under the assumption of practical linear
beamforming treating interference as noise. First, Jorswieck {\em
et al.} investigated the structure of optimal beam vectors
achieving Pareto boundary points of the achievable rate region of
the MISO IC with linear beamforming
\cite{Jorswieck&Larsson&Danev:08SP} and showed that any
Pareto-optimal beam vector at each transmitter is a normalized
convex combination of the ZF beam vector and  matched-filtering
(MF) (i.e., maximal ratio transmission) beam vector in the case of
two users and a linear combination of the channel vectors from the
transmitter to all receivers in the general case of an arbitrary
number of users. The result is extended in
\cite{Mochaourab&Jorswieck:11SP} to general MISO interference
networks with arbitrary utility functions having monotonic
property. Moreover, the parameterization for the Pareto-optimal
beam vector is compressed from $K(K-1)$ complex numbers
\cite{Jorswieck&Larsson&Danev:08SP} to $K(K-1)$ real numbers. In
addition to these  results,  other interesting works for MISO ICs
include the consideration of imperfect CSI
\cite{Lindblom&Larsson&Jorswieck:10WCOM}, shared data
\cite{Bjornson&etal:10SP}, second-order cone programming
\cite{Qiu&Zhang&Luo&Cui:11SP}, etc. Although these works provide
significant theoretical insights into the optimal beam structure
and  parameterization of  Pareto-optimal beam vectors, it is not
easy to use these results to design an optimal beam vector in the
real-world systems, and the beam design problem in the general
case  still remains as a non-trivial  problem practically.

 With a sufficient number of transmit antennas, the simplest beam design method for base station coordination is ZF, which perfectly eliminates interference leakage to undesired receivers. However, it is well known that the ZF method is not optimal in the sense of sum data rate or Pareto-boundary achievability, and there have been several ideas to enhance the ZF beam design method. In the case of multi-user MISO/MIMO broadcast channels, the regularized channel inversion (RCI) \cite{Peel&Hochwald&Swindlehurst:05Com} and the signal-to-leakage-plus-noise (SLNR) method \cite{Sadek&Tarighat&Sayed:07TW} were proposed for this purpose. In particular, the SLNR method maximizes the ratio of signal power (to the desired receiver) to leakage (to undesired receivers) plus noise power, and its solution is given by solving a generalized
eigenvalue problem. The SLNR method can easily be adapted to the
MISO/MIMO IC. Recently, Zakhour and Gesbert rediscovered this method in the context of MISO IC under the name of the virtual signal-to-interference-plus-noise (SINR) method, and have further (and more importantly) shown that this method can achieve any point on the Pareto boundary theoretically, but practically can achieve one uncontrolled point
on the Pareto boundary of the achievable rate region in the case
of two\footnote{It can be shown that the virtual SINR (or SLNR)
method can theoretically achieve any Pareto-optimal point in the
general MISO IC case, too. See the appendix of \cite{Park&Lee&Sung&Yukawa:12Arxiv}.} users
\cite{Zakhour&Gesbert:09WSA},\cite{Zakhour&Gesbert:10WC}.

Another way of generalizing ZF in MISO IC was proposed by relaxing
the ZF leakage constraints to undesired users in
\cite{Shang&Chen&Poor:11IT}, \cite{Zhang&Cui:10SP},
\cite{Lee&Park&Sung:11WCSP}. First, Shang {\em et al.} showed that
all boundary points of the achievable rate region of MISO IC with
single-user decoding can be obtained by linear beamforming
\cite{Shang&Chen&Poor:11IT}, by converting the non-convex weighted
sum rate maximizing precoder design problem into a set of separate
convex problems by taking a lower bound on the achievable rate of
each user under the relaxed ZF (RZF) framework. This method was
further investigated  by Zhang and Cui \cite{Zhang&Cui:10SP}, who
showed that separate rate optimization under the RZF framework
with a set of well-chosen interference leakage levels to undesired
users is Pareto-optimal for MISO ICs in addition to being sum-rate
optimal. In \cite{Lee&Park&Sung:11WCSP}, Lee {\em et al.} extended
the RZF framework to the case of MIMO IC. In this RZF beamforming
framework, each transmitter maximizes its own rate under
interference leakage constraints to undesired receivers. The idea
is based on the simple observation that the ZF beam design method
overreacts to inter-cell interference by completely nulling
out the interference. Most receivers (i.e., mobile stations) that
are affected by inter-cell interference are cell-edge users, and
thus, thermal noise remains even if the inter-cell interference is
completely removed. Thus, it is unnecessary to completely
eliminate the inter-cell interference and it is sufficient to
limit the inter-cell interference to a certain level comparable to
that of the thermal noise. By relaxing ZF interference
constraints, we do not need the condition that the number of
transmit antenna is larger than or equal to that of receivers and
have a larger feasible set yielding a larger rate than that of the
ZF scheme. In this paper, we explore and develop this RZF idea
fully in several aspects to provide a useful design paradigm for
coordinated beamforming (CB) for current and future cellular
networks. The contributions of the paper is summarized as follows:

\noindent $\bullet$ In the MISO IC case, a new structural representation of optimal beam vector for RZF coordinate beamforming is derived.

\noindent $\bullet$ In the MISO IC case, based on the new structural
representation,  the {\it sequential
orthogonal projection combining (SOPC) method} for the RZF beam
design is proposed. In the case of $K=3$, an approximate
closed-form solution is provided.
%(The case of $K=3$ is
%particularly important for the hexagonal cell structure when three
%cells are coordinating their beam vectors.)

\noindent $\bullet$  In the RZF framework, the allowed interference leakage
levels to undesired receivers at each transmitter are design
parameters, and the rate-tuple is controlled by controlling these
interference leakage levels. A centralized algorithm and a fully distributed heuristic algorithm are provided to control the location of the designed
rate-tuple (roughly) along the Pareto boundary of the achievable rate
region.
The controllability of rate is a desirable feature in network
operation since the required data rate of each
transmitter-receiver pair may be different from those of others in
practice, as in an example that one user is a voice user and the
others are high rate data users.

\noindent $\bullet$ Finally, the RZF CB (RZFCB) is extended to the MIMO IC case.
In the MIMO case, a new lower bound on
each user's rate is derived to decompose the beam design
problem into separate problems at different transmitters, and the
projected gradient method \cite{goldstein64} is adopted to solve
the MIMO RZFCB problem.

\noindent \textbf{Notations and Organization}
In this paper, we will make use of standard notational
conventions. Vectors and matrices are written in boldface with
matrices in capitals. All vectors are column vectors. For a matrix
$\Abf$, $\Abf^H$, $\|\Abf\|$, $\|\Abf\|_F$, $\mbox{tr}(\Abf)$, and
$|\Abf|$ indicate the Hermitian transpose, 2-norm,  Frobenius
norm, trace, and determinant of $\Abf$, respectively, and
$\Cc(\Abf)$ denotes the column space of
$\Abf$. $\Ibf_n$ stands for the identity matrix of size $n$
(the subscript is omitted when unnecessary).
$\Pibf_{\Abf}=\Abf(\Abf^H\Abf)^{-1}\Abf^H$ represents the
orthogonal projection onto $\Cc(\Abf)$ and
$\Pibf_{\Abf}^\perp=\Ibf - \Pibf_{\Abf}$.  For matrices $\Abf$ and
$\Bbf$, $\Abf \ge \Bbf$ means that
$\Abf-\Bbf$ is positive semi-definite. $[\abf_1,\cdots,\abf_L]$ or
$[\abf_i]_{i=1}^L$ denotes the matrix composed of vectors
$\abf_1,\cdots,\abf_L$. $\xbf\sim\mathcal{CN}(\mubf,\Sigmabf)$ means that
$\xbf$ is circular-symmetric complex Gaussian-distributed with mean vector $\mubf$ and covariance matrix $\Sigmabf$. ${\mathbb{R}}$,
${\mathbb{R}}_+$, and ${\mathbb{C}}$ denote the sets of real
numbers, non-negative real numbers, and complex numbers,
respectively. For a set $A$, $|A|$ represents the cardinality
of the set.

The remainder of this paper is organized as follows.
The system model and the preliminaries are provided in Section \ref{sec:systemmodel}. In
Section \ref{sec:MISO_IC_RZF}, the RZFCB in MISO ICs is formulated,
and its solution structure and a fast algorithm for RZFCB are
provided. In Section \ref{sec:rate_control}, the rate-tuple control
problem under the RZFCB framework is considered and two approaches
are proposed to control the designed rate-tuple. The RZFCB problem
in MIMO ICs is considered in Section \ref{sec:MIMO_IC}, followed by
conclusions in Section \ref{sec:conclusion}.

%\vspace{-1.5em}
%%%%%%%%%%%%%%%%%%%%%%%%%%%%%%%%%%%%%%%%%%%%%%%%%%%%%%%%
\section{System Model and Preliminaries}
\label{sec:systemmodel}
%%%%%%%%%%%%%%%%%%%%%%%%%%%%%%%%%%%%%%%%%%%%%%%%%%%%%%%%
%\vspace{-0.5em}

In this paper, we consider a multi-user interference channel with
$K$ transmitter-receiver pairs. In the first part of the paper, we restrict
ourselves to the case that the transmitters are equipped with $N$
antennas and each receiver is equipped with one receive antenna
only.  In this case, the received signal at receiver $i$ is given by
\begin{equation} \label{eq:rec_signal}
 y_{i} = \hbf_{ii}^H\vbf_{i} s_i + \sum\limits_{j=1, j \neq i}^K
\hbf_{ij}^H\vbf_j s_j + n_i,
\end{equation}
where $\hbf_{ij}$ denotes the $N\times 1$ (conjugated) channel
vector from transmitter $j$ to receiver $i$, and $\vbf_j$ and
$s_j$ are the $N \times 1$ beamforming vector and the scalar
transmit symbol at transmitter $j$, respectively.
We assume that the transmit symbols are from a Gaussian code book with unit
variance,  the additive noise $n_i$ is from ${\mathcal{CN}}(0, \sigma_i^2)$, and  each transmitter has a transmit power constraint, $\|\vbf_i\|^2 \leq P_i$, $i = 1, \cdots, K$.

The first term on the right-hand side (RHS) of
\eqref{eq:rec_signal} is the desired signal and the second term
represents the sum of interference from  $K-1$ undesired
transmitters. Under single-user decoding at each receiver treating interference as noise, for a given set of
beamforming vectors $\{\vbf_1,\cdots,\vbf_K\}$ and  a channel
realization $\{\hbf_{ij}\}$, the rate of receiver $i$  is given by
\begin{equation} \label{eq:R_iOneUser}
R_i(\vbf_1,\cdots,\vbf_K)
=
 \log\left(1+\frac{|\hbf_{ii}^H\vbf_i|^2}
{\sigma_i^2+ \sum_{j\neq i} |\hbf_{ij}^H\vbf_j|^2} \right).
\end{equation}
Then, for the given channel realization, the achievable rate
region of the MISO IC with transmit beamforming and single-user decoding  is defined as
the union of the rate-tuples that can be achieved by all possible
combinations of beamforming vectors under the power constraints:
\begin{equation} \label{eq:rate_region}
{\mathcal{R}} :=
\hspace{-1.4em}
\bigcup_{\left\{\substack{ \vbf_i:\vbf_i\in{\mathbb{C}}^N, \\
\|\vbf_i\|^2\leq P_i,\  1\leq i\leq K }\right\}}
\hspace{-1.4em}
(R_1(\vbf_1, \cdots, \vbf_K),\ \cdots,\ R_K(\vbf_1, \cdots, \vbf_K)).
\end{equation}
The outer boundary of the rate region ${\mathcal{R}}$ is called the
 $\textit{Pareto boundary}$ of $\Rc$ and
it consists of the rate-tuples for which the rate of any one user cannot
be increased without decreasing the rate of at least one other user \cite{Jorswieck&Larsson&Danev:08SP}.

%The precise definition of a Pareto-optimal point of the achievable
%rate region $\Rc$ is given as follows.
%\begin{definition}
%A rate-tuple $(R_1,\cdots,R_K)$ is $\textit{Pareto-optimal}$ if
%there exists no other rate-tuple $(R_1^\prime,$ $\cdots,
%R_K^\prime)$ such that (s.t.) $(R_1^\prime,\cdots,R_K^\prime) \geq
%(R_1,\cdots,R_K)$ and $(R_1^\prime,\cdots,R_K^\prime) \neq
%(R_1,\cdots,R_K)$ (the inequality is component-wise).
%\end{definition}

At each transmitter, the interference to undesired receivers can be
eliminated completely by ZF CB (ZFCB).
Due to its simplicity and fully distributed nature, there has been extensive research on ZFCB, e.g., \cite{Spencer&Swindlehurst&Haardt:04SP,
Shim&Kwak&Heath&Andrews:08WC,
Somekh&Simeone&BarNess&Haimovich&Shamai:09IT}. The best ZF
beamforming vector at transmitter $i$ can be obtained by solving the
following optimization problem:
\begin{align}
\vbf_i^* =& \mathop{\arg\max}_{\vbf_i \in\ {\mathbb{C}}^{N} }\ \
    \log \left( 1+\frac{|\hbf_{ii}^H\vbf_{i}|^2}{\sigma_i^2} \right)
     \label{eq:ZFCB}  \\
 &\mbox{subject to\ \ \ } |\hbf_{ji}^H \vbf_{i}| =
    0,  ~\forall~ j\ne i ~~~~\mbox{and} \quad \|{\vbf}_{i}\|^2 \le
    P_i.\nonumber
\end{align}
Here, $|\hbf_{ji}^H \vbf_{i}| = 0$ is the ZF leakage constraint at
transmitter $i$ for receiver $j$. If $N \ge K$, the problem
\eqref{eq:ZFCB} has a non-trivial solution and the solution is
given by $\vbf_i^{ZF}=c\Pibf_{[\hbf_{1i},\cdots,\hbf_{i-1,i},\hbf_{i+1,i},\cdots,\hbf_{Ki}]}^\perp
\hbf_{ii}$ for some scalar $c$ satisfying the transmit power
constraint. In this paper, however, we do not assume that $N
\ge K$ necessarily as in the ZF beamforming, but assume that

 {\em (A.1)} In the case of $N \ge K$, $\{\hbf_{ji}, j=1,\cdots,K\}$  are linearly independent for each $i$.
In the case of $N < K$, the element vectors of any subset of $\{\hbf_{ji}, j=1,\cdots,K\}$ with cardinality $N$ are linearly independent for each $i$.

\noindent Assumption  {\em (A.1)}   is almost surely satisfied for randomly realized channel vectors.

%\vspace{-1em}
%%%%%%%%%%%%%%%%%%%%%%%%%%%%%%%%%%%%%%%%%%%%%%%%%%%%%%%%
\section{RZF Coordinated Beamforming in MISO Interference Channels}
\label{sec:MISO_IC_RZF}
%%%%%%%%%%%%%%%%%%%%%%%%%%%%%%%%%%%%%%%%%%%%%%%%%%%%%%%%
%\vspace{-0.5em}

\subsection{Formulation}

Although the ZFCB provides an effective way to handling inter-cell
interference, the ZFCB is not optimal from the perspective of Pareto optimality, i.e., the rate tuples achieved by ZFCB are in the interior of the achievable rate region \cite{Larsson&Jorswieck:08JSAC}.  and requires the condition $N \ge K$. As mentioned before, even with such complete interference nulling, there exists thermal noise at each receiver, and thus, a certain level of interference leakage comparable to the power of thermal noise can be allowed for better performance. In the MISO IC case, the RZF leakage constraint at transmitter $i$ for receiver $j$ is formulated as follows:
\begin{equation}\label{eq:RZF}
|\hbf_{ji}^H \vbf_{i} |^2 \le \alpha_{ji}\sigma_j^2, \quad \forall
i, j \neq i,
\end{equation}
where  $\alpha_{ji} \ge 0$ is a constant\footnote{In the RZF
scheme, $\{\alpha_{ji}, j,i=1,\cdots,K, j\neq i\}$ are system design
parameters that should be designed properly for optimal
performance. The practical significance of the parameterization in terms of the interference leakage levels will be clear in Section \ref{subsec:distributedControl}.} that controls the allowed level of interference
leakage from transmitter $i$ to receiver $j$ relative to the
thermal noise level $\sigma_j^2$  at receiver $j$. When $\alpha_{ji}=0$ for all $j\neq i$, the RZF constraints reduce to the conventional ZF constraints. When $\alpha_{ji} >0$, on the other hand, the ZF constraints are relaxed to yield a larger feasible set for $\vbf_i$ than that associated with the ZF constraints and due to this relaxation the condition $N \ge K$ is not necessary anymore.

Under the RZF framework, the power of interference from undesired transmitters at receiver $i$ is upper bounded as
\begin{equation}\label{eq:interf_power}
\textstyle
\sum_{j=1,j\neq i}^K |\hbf_{ij}^H\vbf_{j}|^2  \le \sum_{j\neq i}
\alpha_{ij}\sigma_i^2 =: \epsilon_i\sigma_i^2.
\end{equation}
Therefore, a lower bound on the rate of user $i$
under RZF is obtained by using \eqref{eq:interf_power} as
\begin{equation}\label{eq:lower_sum_rate}
\log\bigg( 1+\frac{|\hbf_{ii}^H\vbf_i|^2}{\sigma_i^2+\sum_{j\neq i}
|\hbf_{ij}^H\vbf_j|^2} \bigg) \ge \log\bigg( 1+\frac{|
\hbf_{ii}^H\vbf_i |^2}{(1+\epsilon_i)\sigma_i^2} \bigg).
\end{equation}
The lower bound on the rate at each receiver does not depend on the
beamforming vectors of undesired transmitters and thus,
exploiting the RZF constraints, we can convert the intertwined coordinated beam design problem into a set of separate problems for different users based on the lower bound \cite{Shang&Chen&Poor:11IT}. The separate problem for each transmitter based on RZF is  given as follows \cite{Shang&Chen&Poor:11IT,Zhang&Cui:10SP}:

\begin{problem}\label{prob:MISO_RZF_formulation1}
For each transmitter $i \in \{1,\cdots,K\}$,
\begin{eqnarray}
&\underset{\vbf_i}{\mbox{maximize}} ~ &
\log\bigg(1+\frac{|\hbf_{ii}^H\vbf_i|^2}{(1+\epsilon_i)\sigma_i^2}\bigg) \\
&\mbox{subject to} & |\hbf_{ji}^H\vbf_i|^2 \leq
\alpha_{ji}\sigma_j^2,
\qquad \forall j\neq i,  \\
&                 & \|\vbf_i\|^2 \leq P_i.
\end{eqnarray}
\end{problem}

\vspace{0.5em} \noindent Then, due to the monotonicity of the
logarithm, Problem \ref{prob:MISO_RZF_formulation1} is equivalent
to the following problem:

\begin{problem}[The MISO RZFCB problem]
\label{prob:MISO_RZF_formulation2} For each transmitter $i \in
\{1,\cdots,K\}$,
\begin{eqnarray}
&\underset{\vbf_i}{\mbox{maximize}}~ &|\hbf_{ii}^H\vbf_i|^2 \\
&\mbox{subject to} & |\hbf_{ji}^H\vbf_i|^2 \leq \alpha_{ji}\sigma_j^2,  \label{eq:Problem3RZFconst}
\qquad \forall j\neq i,  \\
&            & \|\vbf_i\|^2 \leq P_i.
\end{eqnarray}
\end{problem}

\noindent  From now on, we will consider Problem
\ref{prob:MISO_RZF_formulation2} (the RZFCB problem) and refer to
the solution to Problem \ref{prob:MISO_RZF_formulation2} as the RZF
beamforming vector.

%\vspace{-1em}
%%%%%%%%%%%%%%%%%%%%%%%%%%%%%%%%%%%%%%%%%%%%%%%%%%%%
\subsection{The Optimality and Solution Structure of RZFCB in MISO Interference Channels}
\label{subsec:solution_structure}
%%%%%%%%%%%%%%%%%%%%%%%%%%%%%%%%%%%%%%%%%%%%%%%%%%%%

In this subsection, we will investigate the optimality and
structure of the solution to Problem
\ref{prob:MISO_RZF_formulation2}.  We start with the optimality of
the RZFCB scheme. Without inter-cell interference, it is optimal
for the transmitter to use the MF beam vector with
full transmit power. However, with inter-cell interference,
such a selfish strategy leads to poor performance due to large mutual interference \cite{Larsson&Jorswieck:08JSAC}. Thus, to enhance
the overall rate performance in the network, the beamforming
vector should be designed to be as close as possible to the MF
beam vector without giving too much interference to undesired
receivers, and this strategy is the RZFCB in Problem
\ref{prob:MISO_RZF_formulation2} (or Problem
\ref{prob:MISO_RZF_formulation1} equivalently). The optimality of
the RZFCB is given in the following theorem of Shang \textit{et al.} \cite{Shang&Chen&Poor:11IT} or Zhang and Cui \cite{Zhang&Cui:10SP}.

\begin{theorem}\cite{Zhang&Cui:10SP}\label{theo:pareto_achievability}
Any rate-tuple $(R_1,\cdots, R_K)$ on the Pareto boundary of the
achievable rate region defined in \eqref{eq:rate_region} can be
achieved by the RZFCB if the levels $\{\alpha_{ij}\sigma_i^2,\
\forall i,j\neq i\}$ of interference leakage
 are properly
chosen.
\end{theorem}

\begin{proof}
See Proposition 3.2 in \cite{Zhang&Cui:10SP}.
\end{proof}

\noindent Surprisingly, the separate beam design based on the rate lower bound in Problem \ref{prob:MISO_RZF_formulation2} can
achieve any Pareto-optimal point of the achievable rate region if the interference relaxation parameters are well chosen.\footnote{The beamforming vectors from Problem \ref{prob:MISO_RZF_formulation2} are necessary to achieve any point on the Pareto boundary but not sufficient. Not any choice of parameters $\{\alpha_{ij}\}$ leads to a point on the Pareto boundary. }
It was also shown that Problem 2 and the approach in \cite{Mochaourab&Jorswieck:11SP} are two different approaches to the same multi-objective optimization problem \cite{Vazquez&Neira&Lagunas:12WSA}. 
Due to Theorem \ref{theo:pareto_achievability}, in the MISO IC
case, the remaining problems for the RZFCB are {\it i)} {\em  to
construct an efficient algorithm to solve the RZFCB problem for given
$\{\alpha_{ij}\sigma_i^2,\ \forall i,j\neq i\}$} and {\it ii)} {\em  to
devise a method to design $\{\alpha_{ij}\sigma_i^2,\ \forall
i,j\neq i\}$ for controlling the location of the rate-tuple along
the Pareto boundary of the achievable rate region.}  We will
consider Problem \ref{prob:MISO_RZF_formulation2} for given
$\{\alpha_{ij}\sigma_i^2,\ \forall i,j\neq i\}$ here and will
consider the rate control problem in the next section.

First, we will derive an
efficient algorithm for obtaining a good approximate solution to
Problem \ref{prob:MISO_RZF_formulation2} for given
$\{\alpha_{ij}\sigma_i^2,\ \forall i,j\neq i\}$. To do this, we
need to investigate the solution structure of the RZFCB problem.
Instead of solving Problem \ref{prob:MISO_RZF_formulation1} as in
\cite{Zhang&Cui:10SP} (this becomes complicated due to the
logarithm), we here solve Problem
\ref{prob:MISO_RZF_formulation2}, which is equivalent to Problem
\ref{prob:MISO_RZF_formulation1}.
Note that Problem \ref{prob:MISO_RZF_formulation2} is not a
convex optimization problem since it maximizes a convex cost function
under convex constraint sets instead of minimizing the cost. However,
Problem \ref{prob:MISO_RZF_formulation2} can be made an equivalent
convex problem by exploiting the phase ambiguity of the solution to
Problem \ref{prob:MISO_RZF_formulation2} and making $\hbf_{ii}^H\vbf_i$
real and nonnegative without affecting the value of
$|\hbf_{ii}^H\vbf_i|$ as follows \cite{Bengtsson&Ottersten:99}:
\begin{problem}\label{prob:MISO_RZF_formulation3}
For each transmitter $i \in \{1,\cdots,K\}$,
\begin{eqnarray}
&\underset{\vbf_i}{\mbox{maximize}} ~ &
\hbf_{ii}^H\vbf_i     \\
&\mbox{subject to} & |\hbf_{ji}^H\vbf_i|^2 \leq
\alpha_{ji}\sigma_j^2,
\qquad \forall j\neq i,  \\
&                 & \|\vbf_i\|^2 \leq P_i , \\
&                 & \hbf_{ii}^H\vbf_i \geq 0.  \label{eq:Prob3imag}
\end{eqnarray}
\end{problem}
Here, the constraint \eqref{eq:Prob3imag} implies $\mbox{imag}(\hbf_{ii}^H\vbf_i)=0$  and due to this constraint, maximizing $|\hbf_{ii}^H\vbf_i|^2$ is equivalent to maximizing $\hbf_{ii}^H\vbf_i$.

\begin{lemma} \label{lemma:combination}
Let $\vbf_i^{opt}$ be a solution of the RZFCB problem
(i.e., Problem \ref{prob:MISO_RZF_formulation2}) for transmitter $i$.
Then, $\vbf_i^{opt}$ is represented as follows:
\begin{equation}  \label{eq:JorswieckParam}
\vbf_i^{opt} = c_{ii}\hbf_{ii}+\sum_{j \in \Gamma_i}c_{ji}\hbf_{ji}
\end{equation}
for some $\{c_{ji}\in{\mathbb{C}}: j\in\Gamma_i\cup\{i\}\}$, where
$\Gamma_i :=\{j: |\hbf_{ji}^H\vbf_i^{opt}|^2 =
\alpha_{ji}\sigma_j^2\}$, $\|\vbf_i^{opt}\|^2 = P_i$ for $N\ge K$, and $\|\vbf_i^{opt}\|^2\le P_i$ for $N < K$.
\end{lemma}

\begin{proof} Proof is based on the equivalent formulation in Problem \ref{prob:MISO_RZF_formulation3}.
Since Problem \ref{prob:MISO_RZF_formulation3} is a convex optimization problem,  the optimal solution can be obtained by the Karush-Kuhn-Tucker (KKT) conditions. The Lagrangian of Problem \ref{prob:MISO_RZF_formulation3} for transmitter $i$ is given by
\begin{eqnarray}\label{eq:Lagrangian}
& &\hspace{-1.8em} {\mathcal{L}}(\vbf_i, \lambdabf, \mu, \nu)
= - \hbf_{ii}^H \vbf_i \\
& &\hspace{-1.7em}
+\sum_{j=1, j\neq i}^{K}
  \lambda_j(|\hbf_{ji}^H\vbf_i|^2-\alpha_{ji}\sigma_j^2)
+\mu(\|\vbf_i\|^2-P_i)
-\nu\hbf_{ii}^H\vbf_i, \nonumber
\end{eqnarray}
where $\lambdabf:=\{\lambda_j\ge0: j=1,\cdots,i-1,i+1,\cdots, K\}$ and $\mu, \nu\ge 0$ are real dual variables. With optimal dual variables $\lambdabf^\star$, $\mu^\star$, and $\nu^\star$, the (complex)  gradient of the Lagrangian should be zero at  $\vbf_i^{opt}$, i.e.,
\begin{align}\label{eq:grad}
\mathbf{0} & = \nabla_{\vbf_i^*}{\mathcal{L}}(\vbf_i, \lambdabf^\star,
\mu^\star, \nu^\star)\big|_{\vbf_i=\vbf_i^{opt}}   \\
& =
- \hbf_{ii} + \sum_{j=1,j\neq i}^{K}
\lambda_j^\star\hbf_{ji}\hbf_{ji}^H\vbf_i^{opt}
+ \mu^\star\vbf_i^{opt} - \nu^\star\hbf_{ii}  \nonumber \\
& = - \hbf_{ii} + \sum_{j\in\Gamma_i}
\lambda_j^\star\hbf_{ji}\hbf_{ji}^H\vbf_i^{opt} +
\mu^\star\vbf_i^{opt} - \nu^\star\hbf_{ii}, \nonumber
\end{align}
where $\Gamma_i := \{j: \lambda_j^\star>0 \}$ and $\nabla_{\vbf_i^*}$ is the conjugate Wirtinger gradient. From the
complementary slackness condition, $\lambda_j^\star>0$ only when
$|\hbf_{ji}^H\vbf_i|^2=\alpha_{ji}\sigma_i^2$. Also, from the complementary slackness, we have
$\nu^\star=0$. Otherwise, $\hbf_{ii}^H\vbf_i^{opt}=0$ and thus no rate is provided to user $i$.
Thus, the gradient of the Lagrangian becomes zero if and only if
\begin{equation} \label{eq:Lagrangian_deriv}
\textstyle
\hbf_{ii}
=\left( \mu^\star\Ibf
+\sum_{j\in\Gamma_i}\lambda_j^\star\hbf_{ji}\hbf_{ji}^H\right)
\vbf_i^{opt}.
\end{equation}
If $\Qbf:=(\mu^\star\Ibf
+\sum_{j\in\Gamma_i}\lambda_j^\star\hbf_{ji}\hbf_{ji}^H)$ is singular, then $\vbf_i^{opt}$ exists if and only if  $\hbf_{ii}\in{\mathcal{C}}(\Qbf)$. However, the condition  $\hbf_{ii}\in{\mathcal{C}}(\Qbf)$ does not
occur almost surely for randomly realized channel vectors, which is assumed here.
Therefore, $\Qbf$ should have full rank for the existence of $\vbf_i^{opt}$ and the corresponding
 $\vbf_i^{opt}$ has two different forms according to the optimal dual variable $\mu^\star$.

$i)\ \mu^\star > 0$: This corresponds to the case in
which the transmitter uses full power, i.e.,
$\|\vbf_i^{opt}\|^2=P_i$. In this case, the optimal solution is given by
\begin{eqnarray} \label{eq:lagrangian2}
\textstyle
\vbf_i^{opt} =
\left(
\mu^\star\Ibf+\sum_{j\in\Gamma_i}\lambda_j^\star\hbf_{ji}\hbf_{ji}^H
\right)^{-1}\hbf_{ii}.
\end{eqnarray}
By applying the matrix inversion lemma recursively, it can be shown that
$\vbf_i^{opt}$ is a linear combination of $\{\hbf_{ji}:
~j\in\Gamma_i^\prime:=\Gamma_i\cup\{i\}\}$. Thus, the solution is represented as  \eqref{eq:JorswieckParam}.

$ii)\ \mu^\star = 0$: This case corresponds to the case in
which full power is not used at transmitter $i$. In this case,
$\Qbf= \sum_{j\in{\Gamma}_i}\lambda_j^\star\hbf_{ji}\hbf_{ji}^H$. The
 matrix $\Qbf$ in this case is non-singular if and only if
 $|\Gamma_i|\ge N$ (i.e., $K > N$) under the assumption {\em (A.1)}, and
 the corresponding solution is
given by
\begin{eqnarray} \label{eq:lagrangian3}
\textstyle
\vbf_i^{opt} =
\left(
\sum_{j\in\Gamma_i}\lambda_j^\star\hbf_{ji}\hbf_{ji}^H
\right)^{-1}\hbf_{ii}.
\end{eqnarray}
In this case, $\{\hbf_{ij}, j \in \Gamma_i\}$ alone span
 ${\mathbb{C}}^N$ fully and 
it is therefore clear that the solution is represented as
 \eqref{eq:JorswieckParam}.
Indeed, any subset of
$\{\hbf_{ji},j=1,\cdots,K\}$ with cardinality $N$ forms a full basis
for ${\mathbb{C}}^N$ under the assumption {\em (A.1)} in this case.

Furthermore, when $N \ge K$, $\vbf_i^{ZF}$ is feasible and thus, we
 can always increase power and rate without causing interference to the
 undesired receivers. Therefore,  the optimal solution uses full power,
 i.e., $||\vbf_i^{opt}||^2 =P_i$ when $N \ge K$. On the other hand,
when $N < K$, we can have either $\mu^\star > 0$ ($||\vbf_i^{opt}||^2 =
 P_i$) or $\mu^\star = 0$ ($||\vbf_i^{opt}||^2 < P_i$).
\end{proof}

The solution to RZFCB for a given set of interference
relaxation levels is a linear combination of the desired channel
and a subset of interference channels for which the RZF constraint
\eqref{eq:Problem3RZFconst} is satisfied with equality.
Furthermore,
 it was shown that the interference leakage levels should be designed to make the RZF interference leakage constraints be satisfied tightly  in order to achieve a point on the Pareto boundary \cite{Zhang&Cui:10SP}. In this case, $\Gamma_i = \{1,\cdots,K\}\backslash \{i\}$ and thus, the RZF beam structure  in Lemma \ref{lemma:combination} coincides with  the Pareto-optimal beam structure derived by Jorswieck {\em et al.} in \cite{Jorswieck&Larsson&Danev:08SP}. Now, based on Lemma \ref{lemma:combination}, we present a new useful representation of $\vbf_i^{opt}$ that provides a clear insight into the RZFCB solution and a basis for fast algorithm construction.

\begin{theorem} \label{theo:SuccessiveZeroForcing}
For transmitter $i$, the RZFCB solution can also be expressed as
\begin{equation}\label{eq:solution_structure}
\vbf_i^{opt} = c_0
\frac{\hbf_{ii}}{\|\hbf_{ii}\|} + c_1
\frac{\Pibf_{\Abf_1}^\perp\hbf_{ii}}{\|\Pibf_{\Abf_1}^\perp\hbf_{ii}\|}
+ \cdots +c_{|\widetilde{\Gamma}_i|}
\frac{\Pibf_{\Abf_{|\widetilde{\Gamma}_i|}}^\perp\hbf_{ii}}{\|\Pibf_{\Abf_{|\widetilde{\Gamma}_i|}}^\perp\hbf_{ii}\|},
\end{equation}
where 
$c_j\in{\mathbb{C}}, ~  j = 0, 1,\cdots , |\widetilde{\Gamma}_i|$ and
$\Abf_j$ is constructed recursively as
\begin{equation} \label{theo:SuccessiveZeroForcing_ABFj}
\Abf_j:=[\Abf_{j-1},\ \hbf_{\widetilde{\Gamma}_i(j), i}], ~~~j=1,
\cdots, |\widetilde{\Gamma}_i|.
\end{equation}
Here for convenience we let $\Abf_0$
be an $N \times 0$ 'matrix'.
$\widetilde{\Gamma}_i$ is a
set made by permuting the elements of $\Gamma_i$ according to an arbitrary order, and  $\widetilde{\Gamma}_i(j)$ denotes the $j$-th element of $\widetilde{\Gamma}_i$.
\end{theorem}

\begin{proof}
From Lemma \ref{lemma:combination}, we know that $\vbf_i^{opt}\in
\Cc([\hbf_{ji}]_{j\in\Gamma_i^\prime})$. Proof of the theorem is
given by showing the equivalence of the two subspaces
$\Cc([\hbf_{ji}]_{j\in\Gamma_i^\prime})$ and
$\Cc([\hbf_{ii},\Pibf_{\Abf_{1}}^\perp\hbf_{ii}, \cdots,$
$\Pibf_{\Abf_{|\widetilde{\Gamma}_i|}}^\perp\hbf_{ii}])$.

 {\em Case (i). $|\Gamma_i|(=|\tilde{\Gamma}_i|) \le N-1$: } In this case,
$\{\hbf_{ii},\Pibf_{\Abf_{1}}^\perp\hbf_{ii}, \cdots,
\Pibf_{\Abf_{|\widetilde{\Gamma}_i|}}^\perp\hbf_{ii}\}$ are linearly
independent. This is easily shown  by replacing
$\Pibf_{\Abf_j}^\bot$ with $\Ibf-\Pibf_{\Abf_j}$ and by using the
linear independence of $\{\hbf_{ji}\}_{j\in\Gamma_i^\prime}$. Thus,
the dimension of $\Cc([\hbf_{ii}, \Pibf_{\Abf_{1}}^\perp\hbf_{ii},
\cdots,$ $\Pibf_{\Abf_{|\widetilde{\Gamma}_i|}}^\perp\hbf_{ii}])$ is
$(|\widetilde{\Gamma}_i|+1)$, which is the same as that of
$\Cc([\hbf_{ji}]_{j\in\Gamma_i^\prime})$. Now, consider the
projection of any vector in $\Cc([\hbf_{ii},
\Pibf_{\Abf_1}^\perp\hbf_{ii}, \cdots,
\Pibf_{\Abf_{|\widetilde{\Gamma}_i|}}^\perp\hbf_{ii}])$ onto the
orthogonal complement of $\Cc([\hbf_{ji}]_{j\in\Gamma_i^\prime})$: {
\begin{eqnarray}
& &   \Pibf_{[\hbf_{ji}]_{j\in\Gamma_i^\prime}}^\perp\big( c_0\hbf_{ii}
 +c_1 \Pibf_{\Abf_1}^\perp\hbf_{ii} + \cdots
 +c_{|\widetilde{\Gamma}_i|}
\Pibf_{\Abf_{|\widetilde{\Gamma}_i|}}^\perp\hbf_{ii}\big)
\nonumber \\
&=&   \Pibf_{[\hbf_{ji}]_{j\in\Gamma_i^\prime}}^\perp\big( c_0\hbf_{ii}
 +c_1 (\Ibf-\Pibf_{\Abf_1})\hbf_{ii} + \cdots
 +c_{|\widetilde{\Gamma}_i|}
(\Ibf-\Pibf_{\Abf_{|\widetilde{\Gamma}_i|}})\hbf_{ii}\big)
\nonumber \\
&=&   \Pibf_{[\hbf_{ji}]_{j\in\Gamma_i^\prime}}^\perp\big( c_0\hbf_{ii}
 +c_1 (\Ibf-\Pibf_{\hbf_{\widetilde{\Gamma}_i(1),i}})\hbf_{ii} + \cdots
 +c_{|\widetilde{\Gamma}_i|}
(\Ibf-\Pibf_{[\hbf_{\widetilde{\Gamma}_i(j),i}]_{j=1}^{|\widetilde{\Gamma}_i|}}
)\hbf_{ii}\big) \nonumber \\
&=& \textstyle \Pibf_{[\hbf_{ji}]_{j\in\Gamma_i^\prime}}^\perp\Big(
    \sum_{j=0}^{|\widetilde{\Gamma}_i|} c_j\hbf_{ii}
 - c_1 \Pibf_{\hbf_{\widetilde{\Gamma}_i(1),i}}\hbf_{ii}
 - \cdots
 - c_{|\widetilde{\Gamma}_i|}
\Pibf_{[\hbf_{\widetilde{\Gamma}_i(j),i}]_{j=1}^{|\widetilde{\Gamma}_i|}}
\hbf_{ii}\Big) \nonumber \\
&=& 0. \label{eq:proj_onto_ortho_comple1}
\end{eqnarray}}
By \eqref{eq:proj_onto_ortho_comple1} the orthogonal complement of
$\Cc(\big[\hbf_{ii}, \Pibf_{\Abf_1}^\perp\hbf_{ii}, \cdots,$
$\Pibf_{\Abf_{|\widetilde{\Gamma}_i|}}^\bot\hbf_{ii}\big])$ is
included in that of $\Cc([\hbf_{ji}]_{j\in\Gamma_i^\prime})$, but
$\Cc(\big[\hbf_{ii}, \Pibf_{\Abf_1}^\perp\hbf_{ii}, \cdots,$
$\Pibf_{\Abf_{|\widetilde{\Gamma}_i|}}^\bot\hbf_{ii}\big])$ and
$\Cc([\hbf_{ji}]_{j\in\Gamma_i^\prime})$ have the same dimensions.
Thus, the two orthogonal complements are the same, and hence, the
two subspaces themselves are the same. Consequently, for any
$c_{ii}\hbf_{ii}+\sum_{j \in \Gamma_i} c_{ji}\hbf_{ji}$ with
arbitrary $\{c_{ji} \in {\mathbb{C}}: j\in\Gamma_i^\prime\}$,
there exists some $\{c_j\in{\mathbb{C}}: 0\le j\le
|\widetilde{\Gamma}_i|\}$ s.t.
\[
\vbf_i^{opt}
= c_{ii} \hbf_{ii} + \sum_{j \in{\Gamma_i}} c_{ji} \hbf_{ji}
= c_0 \frac{\hbf_{ii}}{\|\hbf_{ii}\|}
    + \sum_{j=1}^{|\widetilde{\Gamma}_i|}
      c_j\frac{\Pibf_{\Abf_j}^\bot\hbf_{ii}}{\|\Pibf_{\Abf_j}^\bot\hbf_{ii}\|}.
\]

 {\em Case (ii). $|\Gamma_i|\ge  N$:} In this case, both $\{\hbf_{ji}, j\in\Gamma_i^\prime\}$ and
$\{\hbf_{ii},\Pibf_{\Abf_{1}}^\perp\hbf_{ii}, \cdots, \Pibf_{\Abf_{|\widetilde{\Gamma}_i|}}^\perp\hbf_{ii}\}$ span the whole ${\mathbb{C}}^N$. Thus, the claim is trivially satisfied.
\end{proof}

\noindent Theorem \ref{theo:SuccessiveZeroForcing} states that the RZF
solution  is a linear combination of vectors that are obtained by
projecting the desired channel vector onto the orthogonal
complements of a series of subspaces spanned by the channels from
the transmitter to the undesired receivers. Furthermore, the
series of subspaces are obtained by sequentially including one
additional interference channel vector at a time, as shown in
\eqref{theo:SuccessiveZeroForcing_ABFj}. Soon, it will be shown
that, to obtain the RZF solution to Problem
\ref{prob:MISO_RZF_formulation2}, the order of interference
channel inclusion for constructing  $\Abf_j$s in Theorem
\ref{theo:SuccessiveZeroForcing} is determined by the set of
allowed interference levels and the channel realization.

%\vspace{-1em}
%%%%%%%%%%%%%%%%%%%%%%%%%%%%%%%%
\subsection{The Sequential Orthogonal Projection Combining Method
and Closed-Form Solutions}
%%%%%%%%%%%%%%%%%%%%%%%%%%%%%%%%

In this subsection, we propose an efficient beam design method for
RZFCB that successively allocates the transmit power to certain
vectors obtained by sequential orthogonal projection of the
desired channel vector onto monotonically decreasing subspaces.
Furthermore, we provide the closed-form solution to the RZFCB
problem in the two-user case and an approximate closed-form
solution in the three-user case.

To obtain the RZF beamforming
vector under given interference relaxation constraints for a given
channel realization, Problem \ref{prob:MISO_RZF_formulation2}
should be solved. One can use a numerical method \cite{Bland&Goldfarb&Todd:81OR}, as in \cite{Zhang&Cui:10SP}. However, such a method
requires a numerical search  for determining the Lagrange dual
variables satisfying the RZF constraints and the transmit power
constraint. To circumvent such difficulty and to increase the
practicality  of the RZFCB, we exploit Theorem
\ref{theo:SuccessiveZeroForcing} to construct an efficient method
to find the RZFCB solution. Theorem
\ref{theo:SuccessiveZeroForcing} provides us with a very
convenient way of obtaining the RZFCB solution for given
interference leakage levels for a given channel realization; we
only need to find $\widetilde{\Gamma}_i$ and complex coefficients
$\{c_i\}$ in \eqref{eq:solution_structure} for each transmitter.
The idea is based on the fact that  the RZF beamforming vector
should be designed to be as close as possible to the MF beam
vector under the interference leakage constraints for the maximum
rate under RZF, as described in Problem
\ref{prob:MISO_RZF_formulation2}. Hereafter, we will explain how
the coefficients $\{c_i\}$ and the matrices $\{\Abf_i\}$ in
Theorem \ref{theo:SuccessiveZeroForcing} can be obtained to
maximize the rate under the RZF interference and power
constraints. Consider transmitter $i$ without loss of generality.
For the given transmit power constraint $\|\vbf_i\|_2^2 \le P_i$,
it may not be possible to allocate all of the transmit power to the
MF direction $\hbf_{ii}$ because this allocation may violate the
RZF constraints. The rate greedy approach under the
 RZF constraints for a given channel realization is explained as follows.
First, we should  start to allocate
 the transmit power to the direction of $\hbf_{ii}$ by increasing
 $c_0$ with some phase until this allocation hits  one of the RZF constraints with equality, i.e.,  the interference level to one of the undesired receivers reaches
  the allowed maximum exactly. (In the case that the allowed interference
  levels to all undesired receivers are the same, this
 receiver is the receiver whose channel vector has the maximum inner product
 with $\hbf_{ii}$.) The index of this receiver is
 $\widetilde{\Gamma}_i(1)$.  At this point, transmitter $i$ cannot allocate
 the transmit power to the direction $\hbf_{ii}$ anymore since this
 would violate the RZF constraint for receiver $\widetilde{\Gamma}_i(1)$. Since the RZF constraints for other
 undesired receivers are still met with strict inequality,
 transmitter $i$ can still cause interference to the remaining
 receivers. Thus, for the maximum rate under the RZF constraints,
 transmitter $i$ should now start to allocate the remaining power to the
 direction of
 $\Pibf_{\Abf_1}^\perp
\hbf_{ii}$, where $\Abf_1=[\hbf_{\widetilde{\Gamma}_i(1),i}]$,  until
this allocation hits another RZF constraint with equality. The
index of this  receiver is $\widetilde{\Gamma}_i(2)$. (Note that
$\Pibf_{\Abf_1}^\perp \hbf_{ii}$ is the direction of maximizing
the data rate without causing additional interference to receiver
$\widetilde{\Gamma}_i(1)$.) Now, transmitter $i$ cannot cause
interference to receiver $\widetilde{\Gamma}_i(2)$ in addition to
receiver $\widetilde{\Gamma}_i(1)$ anymore. Therefore, at this point,
transmitter $i$ should start to allocate its remaining power to
the next greedy direction $\Pibf_{\Abf_2}^\perp \hbf_{ii}$, where
$\Abf_2=[\hbf_{\widetilde{\Gamma}_i(1),i},
\hbf_{\widetilde{\Gamma}_i(2),i} ]$.
This greedy power allocation
without violating the RZF constraints should be done until either all the
transmit power is used up ($\mu^\star > 0$ in Lemma \ref{lemma:combination}) or we cannot find a new direction that does not cause interference to the users that are already in the set $\tilde{\Gamma}_i$ ($\mu^\star = 0$ in Lemma \ref{lemma:combination}). When $N \ge K$ and transmit power still remains even after hitting all the $K-1$ interference leakage constraints with equality, from then on, all the remaining power should be
allocated to the ZF direction. This coincides with our intuition that  ZF is optimal at a high signal-to-noise ratio (SNR) in the case of $N \ge K$. On the other hand,
when all the transmit power is used up before reaching the remaining interference constraints
with equality, the corresponding remaining interference channel vectors do not
appear in the solution. The final RZF solution is the sum of these
component vectors and has the form in
\eqref{eq:solution_structure}. In this way, the RZFCB solution can
be obtained by combining the sequential projections of the desired
channel vector $\hbf_{ii}$ onto the orthogonal complements of the
subspaces $\Cc(\Abf_1) \subset \cdots \subset
\Cc(\Abf_{|\widetilde{\Gamma}_i|})$. Thus, we refer to this beam
design method as the {\em sequential orthogonal projection combining (SOPC)
 method.\footnote{The rate optimality of the SOPC strategy under the RZF constraints is straightforward to see. Suppose that we are given any beam vector that is a linear combination of $\{\hbf_{ji}\}$, satisfies the RZF interference and power constraints but is not  the SOPC solution. Then, the vector can still be represented in terms of the SOPC basis in Theorem \ref{theo:SuccessiveZeroForcing} and some of the basis component vectors with larger inner product with the MF direction do not satisfy the RZF constraints with equality. Thus,
  the rate can be increased by allocating power from the basis component vector with smaller inner product with the MF direction to the basis component vector with larger inner product with the MF direction until the RZF constraints are satisfied with equality.} } By Theorem  \ref{theo:pareto_achievability}, {\em the SOPC strategy with a well chosen set of interference relaxation levels is Pareto-optimal for MISO $K$-pair interference channels  with single-user decoding.}

 An interesting interpretation of the SOPC strategy is in an analogy with the water-filling strategy. The water-filling strategy distributes power to resource bins according to the effectiveness of each bin, and the power fills into the bin with the lowest noise level (or the most effective bin) first.
Similarly, the SOPC strategy allocates power to the most effective
direction first and then the next most effective direction when
the first direction cannot accommodate  power anymore. This procedure
continues until either the procedure uses up the power or it cannot find a new feasible direction.
  So, the SOPC strategy can be viewed graphically as pouring water on top of a multi-tiered fountain, as
illustrated in Fig. \ref{fig:SOPC_strategy}.
The relationship of the RZFCB/SOPC design and the two-user result by Jorswieck {\em et al.} \cite{Jorswieck&Larsson&Danev:08SP} is explained in Fig.
\ref{fig:SOPC_interpret}.
In the two user case, Jorswieck {\em et al.} have shown that a Pareto-optimal beam vector is a convex combination of the MF beam $\vbf_i^{MF}$ and the ZF beam $\vbf_i^{ZF}$
satisfying the power constraint, i.e., $\vbf_i=\sqrt{P_i}\frac{\lambda_i\vbf_i^{MF}+(1-\lambda_i)\vbf_i^{ZF}}{\|\lambda_i\vbf_i^{MF}+(1-\lambda_i)\vbf_i^{ZF}\|}$,
where $0\le\lambda_i\le 1$.
 Thus, the feasible set of optimal
beam vectors is the arc denoted by $\Fc$ in Fig. \ref{fig:SOPC_interpret}. All the points on this arc can be represented by the sum of
the two vectors in red, and the size of the component vector in
the MF direction is determined by its projection onto
$\Cc(\hbf_{21})$, i.e., the allowed interference level to the
other receiver in the RZF context. Thus, the two-user result by
Jorswieck {\em et al.} can be viewed as a special case of the SOPC
strategy when the number of users is two. The key difference is the parameterization; $\alpha_{12}$ and $\alpha_{21}$ are the parameters in the RZF framework whereas  the linear combining coefficients $\lambda_1$ and $\lambda_2$ are the parameters in \cite{Jorswieck&Larsson&Danev:08SP}.

\begin{figure}[t]
\centerline{
\begin{psfrags}
    \psfrag{a}[c]{$\hbf_{ii}$} %
    \psfrag{b}[c]{$\ \Pibf_{\Abf_1}^\bot\hbf_{ii}$} %
    \psfrag{c}[c]{$\ \dots$} %
    \psfrag{d}[c]{$\ \ \Pibf_{\Abf_K}^\bot \hbf_{ii}$} %
    \psfrag{t1}[l]{$P_i^{(0)}$} %
    \psfrag{t2}[l]{$P_i^{(1)}$} %
    \psfrag{t3}[l]{$P_i^{(K-1)}$} %
    \psfrag{bd}[c]{{\small ZF}} %
    \psfrag{mf}[c]{{\small MF}} %
    \scalefig{0.4}\epsfbox{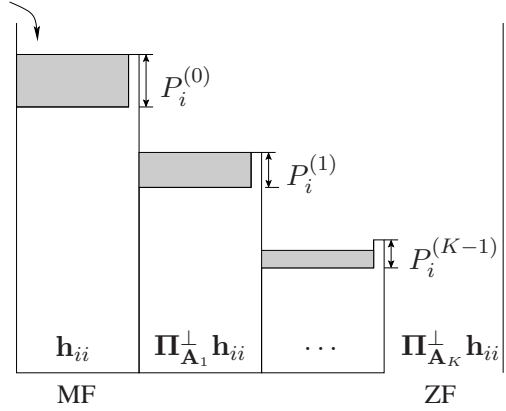}
    \end{psfrags}} \caption{The SOPC strategy in the case of $N \ge K$:
    Water-pouring on a
multi-tiered fountain.} \label{fig:SOPC_strategy}
\end{figure}

\begin{figure}[t]
\centerline{
 \begin{psfrags}
    \psfrag{fc}[c]{${\mathcal{F}}$} %
    \psfrag{coh21}[c]{{\small $\Cc(\Pibf_{\hbf_{21}}^\bot\hbf_{11} )$}} %
    \psfrag{ch11}[c]{{\small $\Cc(\hbf_{11})$}} %
    \psfrag{r}[c]{{\small $\sqrt{P_1}$}} %
    \psfrag{as2}[c]{{\small $\sqrt{\alpha_{21}\sigma_2^2}$}} %
    \psfrag{ch21}[l]{\!\!\!{\small $\Cc(\hbf_{21})$}} %
    \psfrag{opt}[c]{$\vbf_1^{opt}$} %
    \scalefig{0.40}\epsfbox{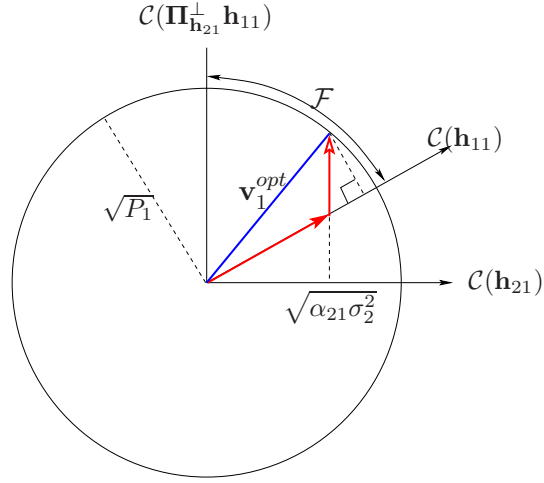}
    \end{psfrags}} \caption{SOPC interpretation of the two-user result.} \label{fig:SOPC_interpret}
\end{figure}

Now consider the detailed implementation of the SOPC method.
Before considering the general case of an arbitrary number $K$ of
users, we consider simple two-user and three-user cases.  Here, we
restrict the combining coefficients $\{c_i\}$ to the set of real
numbers. It will shortly be  shown that the performance loss
caused by restricting  $\{c_i\}$ to real numbers is negligible.
Furthermore, it is the optimal solution of the RZFCB when $K=2$.
For simplicity, we only provide the solution for transmitter 1.
The solutions for other transmitters can be obtained in a similar
way.

%%%%%%%%%%%%%%%%%%%%%%%%%%%%%%%%%%%%%%%%%%%%%%%%
%% K = 2 Closed-Form Solution
\begin{proposition} \label{prop:sopck2}
The closed-form SOPC solution  in the two-pair MISO IC case  is
given by
\begin{align} \label{eq:sopc_K2}
\vbf_1 =
\left\{
\begin{array}{ll}
\sqrt{P_1}\vbf_1^{MF}, ~~ & \text{if} ~~ P_1
\le \frac{\alpha_{21}\sigma_2^2}{|\hbf_{21}^H\vbf_1^{MF}|^2}, \\
\xi_0\vbf_1^{MF}+\xi_1\vbf_1^{ZF},   &  \text{otherwise},
\end{array}
\right. \nonumber
\end{align}
where $\vbf_1^{MF}=\frac{\hbf_{11}}{\|\hbf_{11}\|}$,
$\vbf_1^{ZF}=\frac{\Pibf_{\hbf_{21}}^{\perp}\hbf_{11}}{\|\Pibf_{\hbf_{21}}^{\perp}
\hbf_{11}\|}$, $\xi_0 =
\sqrt{\frac{\alpha_{21}\sigma_2^2}{|\hbf_{21}^H\vbf_1^{MF}|^2}}$,
and $\xi_1 = -\rho\xi_0+\sqrt{P_1-\xi_0^2(1-\rho^2)}$. Here, $\rho
=
(\vbf_1^{MF})^H\vbf_1^{ZF}=\|\Pibf_{\hbf_{21}}^\bot\hbf_{11}\|/\|\hbf_{11}\|
\in {\mathbb{R}}_+$.
\end{proposition}

\vspace{0.3em}
\textit{Proof:}
Proof of Proposition \ref{prop:sopck2} can be found in \cite{Park&Lee&Sung&Yukawa:12Arxiv}.
$\hfill\blacksquare$

%%%%%%%%%%%%%%%%%%%%%%%%%%%%%%%%%%%%%%%%%%%%%%%%
%% K = 3 Closed-Form Solution

\noindent Now, we consider the case of $K=3$. This case is particularly
important when the hexagonal cell structure is used and three
cells are coordinating their beam vectors.  In the case of $K=3$,
the solution  can have six different forms depending on the
transmit power and channel realization. We will provide the
closed-form solution under the real coefficient restriction for transmitter 1 in the case that  the
interference leakage to receiver $3$ reaches the allowed level
before the interference leakage to receiver $2$ reaches the
allowed level. (For this, we should first take inner products
$\langle\hbf_{21}, \hbf_{11}\rangle$ and $\langle\hbf_{31},\hbf_{11}\rangle$
and compare the ratio of their magnitudes with some threshold. The solutions
of the other case and of other users can be derived in the same
manner.)

\begin{proposition}  \label{prop:sopck3}
For $K=3$ and
$\frac{|\hbf_{31}^H\vbf_1^{MF}|^2}{|\hbf_{21}^H\vbf_1^{MF}|^2} \geq
\frac{\alpha_{31}\sigma_3^2}{\alpha_{21}\sigma_2^2 }$, the
closed-form SOPC solution with the restriction to real coefficients at transmitter $1$ is given in \eqref{eq:sopc_k3}.
\begin{figure*}
\hspace{-12em}
\begin{equation} \label{eq:sopc_k3}
\vbf_1 = 
\left\{
\begin{array}{ll}
\sqrt{P_1}\vbf_1^{MF}, \hspace{-4.5em} &
\text{if} ~~ P_1\in\Psi_1:=\{P_1 \in {\mathbb{R}}^+:\sqrt{P_1} \le \beta_0\}, \\
\beta_0\vbf_1^{MF} + \beta_1 \frac{\Pibf_{\hbf_{31}}^\perp\hbf_{11}}
{\|\Pibf_{\hbf_{31}}^\perp \hbf_{11}\|},  &
\mbox{if} ~~ P_1 \in \Psi_2:=  
\Big\{P_1 \in {\mathbb{R}}^+: \sqrt{P_1} > \beta_0, \\
& \hspace{7.5em}
\left|\hbf_{21}^H\Big(\beta_0\vbf_1^{MF} +
\beta_1\frac{\Pibf_{\hbf_{31}}^\perp \hbf_{11}}{\|\Pibf_{\hbf_{31}}^\perp \hbf_{11}\|}\Big)\right|^2 \le \alpha_{21}\sigma_2^2\Big\}, \\
\beta_0\vbf_1^{MF} +
\beta_1^\prime\frac{\Pibf_{\hbf_{31}}^\perp\hbf_{11}}
{\|\Pibf_{\hbf_{31}}^\perp \hbf_{11}\|} + \beta_2 \vbf_1^{ZF}, &
\mbox{if} ~~
P_1 \in \Psi_3:=\Big\{P_1 \in {\mathbb{R}}^+:  \\
& \hspace{7.5em} 
\Big|\hbf_{21}^H\Big(\beta_0\vbf_1^{MF} +
\beta_1\frac{\Pibf_{\hbf_{31}}^\perp
\hbf_{11}}{\|\Pibf_{\hbf_{31}}^\perp \hbf_{11}\|}\Big)\Big|^2 >
\alpha_{21}\sigma_2^2 \Big\}.
\end{array}
\right.
\end{equation}
% \hrule
\end{figure*}
In \eqref{eq:sopc_k3}, $\beta_0$, $\beta_1$, $\beta_1^\prime$ and $\beta_2$ are given by
$\beta_0 = \frac{\sqrt{\alpha_{31}\sigma_3^2}}{|\hbf_{31}^H\vbf_1^{MF}|}$,
$\beta_1 = -a\beta_0 + \sqrt{P_1 - (1-a^2)\beta_0^2}$,
$\beta_1^\prime =
\frac{1}{c}
{(-d \beta_0+\sqrt{d^2\beta_0^2-c(b\beta_0^2-\alpha_{21}\sigma_2^2)})}$,
and
$\beta_2 = -(f\beta_0+e\beta_1^\prime) +
\sqrt{(f\beta_0+e\beta_1^\prime)^2 -
(2a\beta_0\beta_1^\prime+\beta_0^2+\beta_1^{\prime 2}-P_1)}$,
where $a:=\mbox{Re}\{\langle\vbf_1^{MF},\Pibf_{\hbf_{31}}^\perp
\hbf_{11}/||\Pibf_{\hbf_{31}}^\perp
\hbf_{11}||\rangle\}=||\Pibf_{\hbf_{31}}^\perp \vbf_1^{MF}||$,
$b =|\hbf_{21}^H \vbf_1^{MF}|^2$,
$c = \Big|\hbf_{21}^H
\frac{\Pibf_{\hbf_{31}}^\perp \hbf_{11}}{\|\Pibf_{\hbf_{31}}^\perp
\hbf_{11}\|}\Big|^2$,
$d=$ $\mbox{Re}\Big\{(\hbf_{21}^H\vbf_1^{MF})^* \Big(\hbf_{21}^H
\frac{\Pibf_{\hbf_{31}}^\perp\hbf_{11}}{\|\Pibf_{\hbf_{31}}^\perp
\hbf_{11}\|}\Big)\Big\}$,
$e=\frac{|\hbf_{21}^H\hbf_{11}|^2}{\|\hbf_{21}\|^2}$, and $f =$
$(\vbf_1^{MF})^H\vbf_1^{ZF}$.\end{proposition}

\vspace{0.3em}
\textit{Proof:}
Proof of Proposition \ref{prop:sopck3} can be found in \cite{Park&Lee&Sung&Yukawa:12Arxiv}.
$\hfill\blacksquare$

\noindent In the case of $K > 3$, it is cumbersome to
distinguish all possible scenarios for deriving an explicit SOPC
solution. Thus, we propose  an algorithm implementing the SOPC
strategy with real combining coefficients in Table \ref{table:SOPC}. In the general case of $K > 3$,
the implementation of the SOPC algorithm can  be simplified  by the known result in the Kalman filtering theory, provided in the following lemma.

\begin{lemma}[Sequential orthogonal projection \cite{Lototsky:book}]
\label{lemma:sequential_orthogonal} Let $\Hc$ be a Hilbert space
with norm $\| \cdot \|$ and inner product $\langle\cdot,\cdot\rangle$.
Consider $\xbf\in \Hc$ and a closed linear subspace $\Abf_j$ of
$\Hc$. For some $\ybf\in \Hc$ but $\ybf\not\in\Abf_j$,  the
following equality holds
\begin{align}
\Pibf_{\Abf_{j+1}}\xbf
&= \Pibf_{[\Abf_j, \ybf]}\xbf   \\
&=
\Pibf_{\Abf_j}\xbf + \frac{\langle\xbf-\Pibf_{\Abf_j}\xbf,\
\ybf-\Pibf_{\Abf_j} \ybf\rangle}
{\|\ybf-\Pibf_{\Abf_j}\ybf\|^2}(\ybf-\Pibf_{\Abf_j}\ybf). \nonumber
\end{align}
\end{lemma}

\vspace{0.5em} \noindent Since we need to compute
$\Pibf_{\Abf_j}^\perp \hbf_{ii} = (\Ibf -\Pibf_{\Abf_j}) \hbf_{ii}$ in
the SOPC algorithm, Lemma \ref{lemma:sequential_orthogonal} can be
applied recursively by exploiting the fact $\Abf_{j} = [\Abf_{j-1},
\hbf_{\widetilde{\Gamma}(j),i}]$.  Thus, we only need to compute
$\Pibf_{\Abf_{j-1}}\hbf_{\tilde{\Gamma}(j),i}$ for each $j \in
\{1,2,\cdots, K\}$.
The proposed algorithm in Table \ref{table:SOPC} computes the direction and size of the component vector for SOPC directly in each step.

\begin{table}[t]
\centering\caption{The sequential orthogonal projection combining
algorithm.} \label{table:SOPC}
\begin{tabular}{p{330pt}}

\hline

Given channel realization $\{\hbf_{ji}, i, j=1,\cdots, K\}$, pre-determined interference levels $\{\alpha_{ji}\sigma_j^2:$ $i, j=1,\cdots, K, j\neq i\}$, and maximum transmit power $\{P_i:i=1,\cdots,K\}$, perform the following procedure at each transmitter $i \in \{1,\cdots,K\}$.

\textbf{Initialization:} $\vbf_i={\mathbf{0}}$,
$\Abf = \emptyset$,
$\Phi_i= \{1, \cdots, i-1, i+1, \cdots, K\}$, and $k=1$.

\textbf{While} $\ k \le \min(N, K)$,

~~ 1. Let $\ubf:=\frac{\Pibf_A^\bot \hbf_{ii}}{\|\Pibf_A^\bot\hbf_{ii}\|}$.

~~ 2. $\mu_p$ is a positive solution of
$\| \vbf_i + \mu_p \ubf \|_2^2 = P_i$, i.e.,
\[
\mu_p := -\rho_p + \sqrt{\rho_p^2 - (\|\vbf_i\|_2^2-P_i)}
\]
~~~~~ where $\rho_p = Re(\ubf^H\vbf_i)$.

~~ 3. $\mu_j$ is a positive solution of $|\hbf_{ji}^H(\vbf_i+\mu_j \ubf)|^2 = \alpha_{ji}\sigma_j^2$

~~~~~ for each $j\in \Phi_i$, i.e.,
\[
\mu_j := \tfrac{-\rho_j+\sqrt{\rho_j^2-|\hbf_{ji}^H\ubf|^2\cdot
(|\hbf_{ji}^H\vbf_i|^2-\alpha_{ji}\sigma_j^2) }}{|\hbf_{ji}^H\ubf|^2}
\]
~~~~~ where $\rho_j=Re(\vbf_i^H\hbf_{ji}\hbf_{ji}^H\ubf)$.

~~ 4. Obtain $\mu_j^* = \min\limits_{j\in\Phi_i} \{\mu_j\}$ and
$j^* = \mathop{\arg\min}\limits_{j\in\Phi_i} \{\mu_j\}$.

~~ 5. If $\mu_p > \mu_j^*$,
$\vbf_i = \vbf_i + \mu_j^* \ubf$, $\Abf = [\Abf, \hbf_{j^* i}]$,
$\Phi_i=\Phi_i \backslash\{j^*\}$,

~~~~~ $k=k+1$, and go to step 1.

~~~~~ If $\mu_p \le \mu_j^*$,
$\vbf_i = \vbf_i + \mu_p \ubf$. Terminate iteration.

\textbf{end}
\\
\hline
\end{tabular}
\end{table}

\begin{figure}[t]
\centerline{ \scalefig{0.5}
\epsfbox{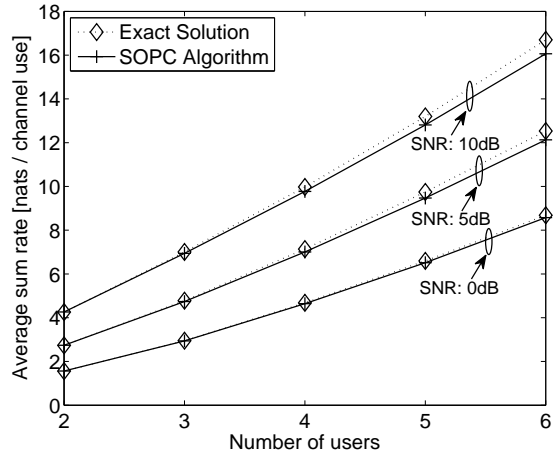} } \vspace{-1.5em} \caption{Average sum
rates of the exact RZFCB solution and the proposed SOPC algorithm with real coefficients.
 (Here, $N=K$ and the average sum rate is obtained over $50$ i.i.d.
channel realizations.)} \label{fig:apsm_wf}
\end{figure}

The proposed SOPC solution based on real coefficients is a sub-optimal solution to the RZFCB problem in the case of  $K\ge 3$. However, the performance loss between the optimal RZFCB (or exact SOPC) beamforming vector and the proposed SOPC solution based on real coefficients is insignificant for a wide range of meaningful SNR values, as seen in Fig. \ref{fig:apsm_wf}. Thus, practically,  the proposed SOPC solution can be used with negligible performance loss. Note  that the
necessary computations for the proposed SOPC solution are a few inner product
and square root  operations and the complexity of the SOPC method is
simply $O(N)$, where $N$ is the number of transmit antennas at the
transmitter. The proposed SOPC method reduces computational complexity to obtain an RZF solution by order of  hundreds
when compared to
 the ellipsoid method for the RZFCB solution used in \cite{Zhang&Cui:10SP}, as shown in Fig. \ref{fig:complexity}, and the solution procedure can easily be  programmed in a real hardware.

\begin{figure}[htbp]
\centering
\scalefig{0.5}\epsfbox{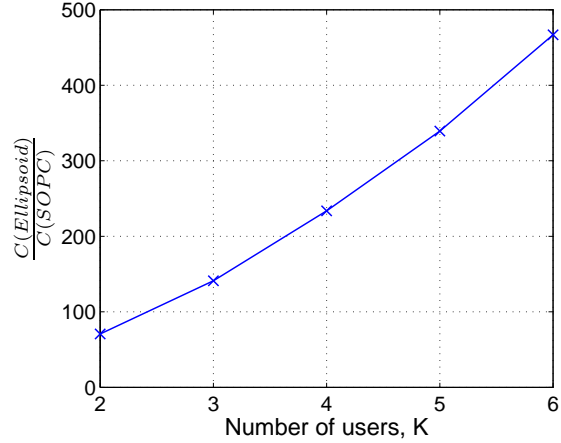}
\caption{Computational complexity for RZFCB beam design: Ellipsoid method \cite{Bland&Goldfarb&Todd:81OR,Zhang&Cui:10SP} versus SOPC ($N=K$ and SNR=5 dB)}.
\label{fig:complexity}
\end{figure}

%\vspace{-0.5em}
%%%%%%%%%%%%%%%%%%%%%%%%%%%%%%%%%%%%%%%%%%%%%%%%%%%%%%%%%%%%%%%%%%%%
\section{Rate-tuple Control}
\label{sec:rate_control}
%%%%%%%%%%%%%%%%%%%%%%%%%%%%%%%%%%%%%%%%%%%%%%%%%%%%%%%%%%%%%%%%%%%%

In the previous section, we  provided an $O(N)$-complexity algorithm
to solve the RZFCB problem for a given set $\{\alpha_{ji}\}$ of
interference relaxation parameters. Now, we  consider how to
design these parameters. We first provide a centralized approach
to determine $\{\alpha_{ji}\}$ with the aim of controlling the
rate-tuple along the Pareto boundary of the achievable rate region
and then a fully-distributed heuristic approach that exploits the
parameterization in terms of interference relaxation levels in
RZFCB and is able to control the rate-tuple location roughly along
the Pareto boundary of the achievable rate region.

%\vspace{-1em}
\subsection{A Centralized Approach}

By Theorem \ref{theo:pareto_achievability}, with a set of well
chosen allowed interference leakage levels, the RZFCB can achieve
any Pareto-optimal point of the rate region. However, the problem
of designing the interference leakage levels $\{\alpha_{ji}\}$ in
the network remains.  Under the RZFCB framework, in
\cite{Zhang&Cui:10SP}, a necessary condition for the interference
relaxation parameters at each receiver to achieve a Pareto-optimal point
was derived. Based on the necessary condition, the authors proposed an
iterative algorithm that updates the interference relaxation
parameters. Although the algorithm in \cite{Zhang&Cui:10SP} is
applicable to general $K$-user MISO interference channels, it
cannot control the rate-tuple location on the Pareto boundary
to which the algorithm converges. To control the rate-tuple to an arbitrary point along the Pareto-boundary of the achievable rate region, we here apply the utility function based approach in \cite{Jorswieck&Larsson:08ICASSP} to the RZF parameterization in terms of interference leakage levels.
Exploiting the fact that the RZFCB can achieve any Pareto-boundary point by adjusting $\{\alpha_{ji}\}$,
we convert the problem of finding a desired point on the
Pareto boundary of the achievable rate region into that of finding
an optimal point of the following optimization problem: {\small 
\begin{eqnarray}
&\max\limits_{\{\alpha_{ji}\}}\  ~~& u \big(R_1(\{\vbf_i^{RZF}(\{\alpha_{ji}\})\}), \cdots, R_K(\{\vbf_i^{RZF}(\{\alpha_{ji}\})\}) \big), \nonumber \\
&\mbox{subject to}~~& |\hbf_{ji}^H\vbf_i|^2\le \alpha_{ji}\sigma_j^2, \quad \forall i, ~j\neq i, \label{eq:newOptProbCentral} \\
&                   & \|\vbf_i\|^2\le P_i, \quad\quad\quad~\forall i, \nonumber
\end{eqnarray}}
where $u(R_1, \cdots, R_K)$ is the desired utility function and several
examples include the weighted sum rate $u(R_1,\cdots,R_K)$ $= \sum w_i
R_i$,
where $w_i\ge 0$ and $\sum w_i = 1$,  the Nash bargaining point  $u(R_1,\cdots,R_K)$
$=\prod_{i=1}^{K} (R_i-R_i^{NE})$, where
$R_i^{NE}=\log_2\left(1+\frac{|\hbf_{ii}^H \vbf_i^{MF}|^2}{\sigma_i^2+\sum_{j\neq
i}|\hbf_{ij}^H \vbf_j^{MF} |^2} \right)$, and the egalitarian point
$u(R_1,\cdots,R_K) = \min(R_1,\cdots, R_K)$
\cite{Jorswieck&Larsson:08ICASSP}.
 The optimization \eqref{eq:newOptProbCentral} can be solved
by an alternating optimization technique. That is, we fix all other
$\alpha_{ji}$'s except one interference relaxation parameter and
update the unfixed parameter so that the utility function is
maximized.  After this update, the next $\alpha_{ji}$ is picked for update. This procedure continues until converges. The
proposed algorithm is described in detail in Table
\ref{table:alpha_algorithm}. For a given utility function $u(R_1,\cdots,R_K)$, the RZF beam vectors $\{\vbf_i^{RZF}\}$ can be obtained as functions of $\{\alpha_{ji}\}$ by the SOPC method, the rate-tuple can be computed as a function of $\{\vbf_i^{RZF}\}$ by \eqref{eq:R_iOneUser}, and finally the utility function value can be computed as a function of $(R_1,\cdots,R_K)$. Thus, the utility value as a function of $\{\alpha_{ji}\}$ can be computed very efficiently by the SOPC method for the proposed centralized algorithm, and this fact makes it easy to apply a numerical optimization method such as the interior point method to the per-iteration optimization in Table \ref{table:alpha_algorithm}.

\begin{figure}[ht]
\centering
\scalefig{0.5}\epsfbox{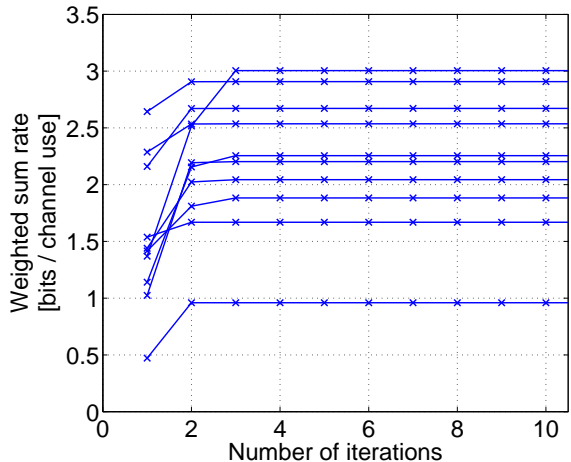}
\caption{ Convergence of the proposed centralized approach ($K=N=2$, $P_i=\sigma_i=1$ for $i=1,2$).}
\label{fig:10_convergence}
\end{figure}

Due to the non-convexity of utility functions w.r.t.
$\{\alpha_{ji}\}$, the convergence of the proposed algorithm to
the global optimum is not guaranteed, but the proposed algorithm
converges to a locally optimal point by the monotone convergence
theorem since  the utility function is upper bounded and the
proposed algorithm yields a monotonically increasing sequence of
utility function values. Furthermore, the proposed algorithm is
also stable by the monotone convergence theorem.
 Fig. \ref{fig:10_convergence} shows the convergence behavior of the proposed utility function based algorithm for 10 different channel realizations when $K=N=2$, $P_i=\sigma_i=1$ ($i=1,2$), and $u(R_1,R_2)=2R_1 + R_2$. It is seen in the figure that the algorithm converges in a few iterations in most cases. Fig. \ref{fig:learning_curves} shows the convergence behavior of several known rate control algorithms for the same setting as in Fig. \ref{fig:10_convergence} for one channel realization.  The considered three algorithms converge to the same value eventually in this case.
It is also seen in Figs. \ref{fig:Pareto1} (a) and (b) that the proposed centralized algorithm yields desired points on the Pareto boundary although it is not theoretically guaranteed.

\begin{table}[!tp]
\caption{A centralized algorithm for determining
$\{\alpha_{ji}\}$.} \label{table:alpha_algorithm} \centering
\begin{tabular}{p{330pt}}
\hline \vspace{-0.8em}

 For given channel realization $\{\hbf_{ji}, i, j=1,\cdots, K\}$,
noise power $\{\sigma_i^2, i=1,\cdots,K\}$, and a utility function
 $u(\{R_i\})$, perform the following procedure to determine interference
 leakage levels $\{\alpha_{ji}\}$.

\vspace{0.5em}
\textbf{Initialization:}
$\{\alpha_{ji}^1=0, ~ i,j = 1,\cdots,K, j\neq i\}$,
$\{R_i^0 = 0, ~i=1,\cdots, K\}$, $\epsilon>0$, and $l=1$.

\vspace{0.5em}
\textbf{while}
$\left|
u(\{\alpha_{ji}^l\} ) - u(\{\alpha_{ji}^{l-1}\})\right| > \epsilon $

\vspace{0.5em}

    \begin{itemize}
                \item[] $l = l+1;$

        \item[] \textbf{for}~ $i=1,\cdots, K$,

        \item[] \hspace{1em} \textbf{for}
        ~ $j=1,\cdots,i-1, i+1, \cdots, K$,

        \item[] \hspace{2.6em}
         $\alpha_{ji}^l ~~=
         \underset{0\le \alpha_{ji} \le P_i|\hbf_{ji}^H\vbf_i^{MF}|^2, }
         {\arg\max}\
         u\Big(\big\{R_k^l(\{\vbf_i^{RZF}(\{\alpha_{ji}^l\})\})\big\}\Big)$

        \item[] \hspace{1em} \textbf{end}

        \item[] \textbf{end}
   \end{itemize}
 \textbf{end}
 \\ \hline
\end{tabular}
\end{table}

\begin{figure}[H]
\centering
\scalefig{0.5}\epsfbox{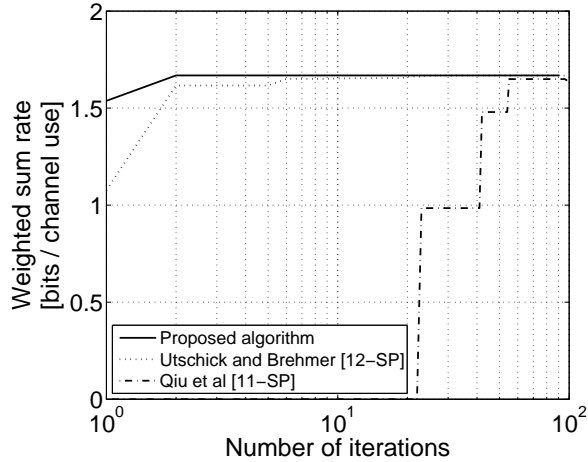}
\caption{Convergence of several known  algorithms ($K=N=2$,
 $P_i=\sigma_i=1$ for $i=1,2$)}.
\label{fig:learning_curves}
\end{figure}

%\vspace{-1em}
%%%%%%%%%%%%%%%%%%%%%%%%%%%%%%%%%%%%%%%%%%%%%%%%%%%%%%%
\subsection{A Distributed Heuristic Approach and Practical Considerations}
%%%%%%%%%%%%%%%%%%%%%%%%%%%%%%%%%%%%%%%%%%%%%%%%%%%%%%%
\label{subsec:distributedControl}

The proposed centralized algorithm in the previous subsection  requires central processing with the knowledge of all
$\{\hbf_{ji}: i,j=1,2,\cdots,K\}$ and
$\{\sigma_{j}^2: j=1,2,\cdots,K\}$. This reduces the practicality
of the centralized approach when communication among the base stations is limited or experiences large delay as in real systems. Note that the RZFCB framework in Problem
\ref{prob:MISO_RZF_formulation2} itself is distributed.
Transmitter $i$ only needs to know $\{\hbf_{ji}, j=1,2,\cdots,K\}$
and $\{\sigma_j^2, j=1,2,\cdots,K\}$ and needs to control
$\{\alpha_{1i},\cdots,\alpha_{i-1,i},\alpha_{i+1,i},
\alpha_{Ki}\}$. In the RZF framework, heuristic rate control is possible with the
knowledge of $\{\hbf_{ji}, j=1,2,\cdots,K\}$ and $\{\sigma_j^2,
j=1,2,\cdots,K\}$ at transmitter $i$.
For fully distributed CB operation with limited inter-base station communication, instantaneous information such as the channel vectors should not be exchanged since inter-base station communication delay is typically larger than the channel coherence time for mobile users. One possible way to roughly control the rate-tuple in the network is to design a table composed of sets of interference relaxation parameters, as in the right side of Fig. \ref{fig:Pareto1}, based on the channel statistics.
When the transmitters form a
coordinating cluster, they can negotiate their rates based on
the requests from their receivers for a communication session. In this phase, one set of interference relaxation levels from the table is picked,  shared among the base stations, and used during the communication session.

Heuristic guidelines to design the parameter table are based on the RZF parameterization itself.
Note that $\epsilon_i =
\sum_{j\ne i} \alpha_{ij}$ in \eqref{eq:interf_power} is the
additional interference power relative to thermal noise power
$\sigma_i^2$ at receiver $i$ and $\epsilon_i =1$ means
that the SINR of receiver $i$ is lower than the SNR of the same
receiver by 3dB. Thus, the designed interference level should not be too high compared to the thermal noise level.
 Furthermore, to (roughly) obtain corner points
of the Pareto boundary of the rate region, another heuristic idea
works. One transmitter should use a nearly MF beam vector, and the
rest of the transmitters should use nearly ZF beam vectors.
More systematic ways based on vast computer simulation can be
considered to design the parameter table. One possible way is as
follows. We first generate a set of channel vectors randomly
according to the channel statistics. For this realized channel
set, we obtain graphs of interference relaxation parameters on the
Pareto boundary. The process is repeated over many different
channel realizations and the best fitting graphs are obtained from
the graphs of interference relaxation parameters of different
channel realizations by some regression model. Finally, the table
is constructed by selecting some points in the best fitting
graphs.  The parameter table in the right side of Fig.
\ref{fig:Pareto1} is obtained in this manner for $K=N=2$
when each element of channel vector is i.i.d. zero-mean
complex Gaussian distributed with unit variance and the SNR is 0
dB. Figures \ref{fig:Pareto1} (a) and (b) show the rate control
performance of the parameter table designed in this manner for two
different channel realizations. It is seen that the heuristic
method performs well; the five rate points are all near the Pareto
boundary for each figure.

\begin{figure*}[t]
\begin{minipage}{0.8\textwidth}
\centerline{
\SetLabels
\L(0.25*-0.1) (a) \\
\L(0.76*-0.1) (b) \\
\endSetLabels
\leavevmode
%\ShowGrid
\strut\AffixLabels{
\scalefig{0.5}\epsfbox{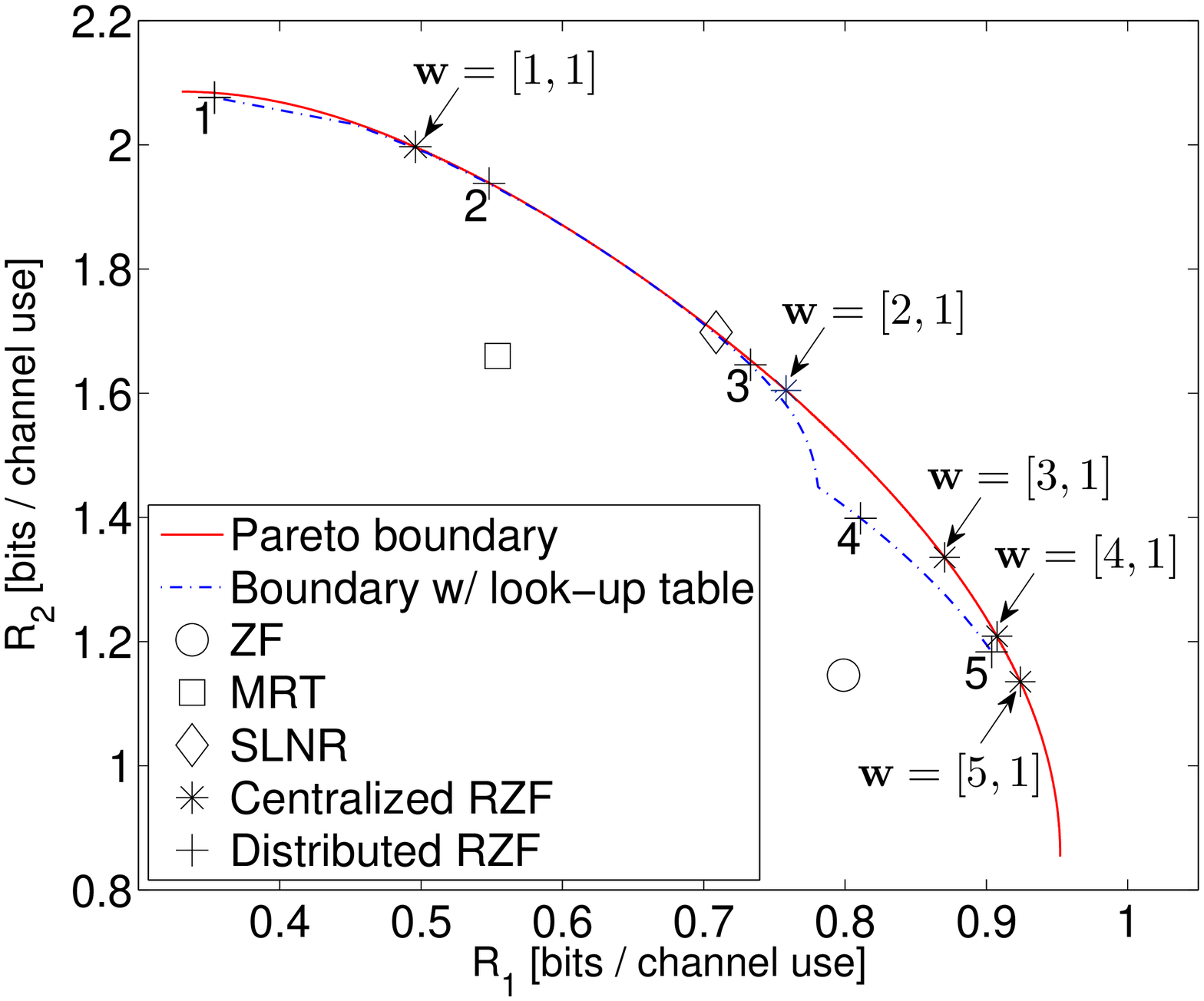}
\hspace{-2em}
\scalefig{0.5}\epsfbox{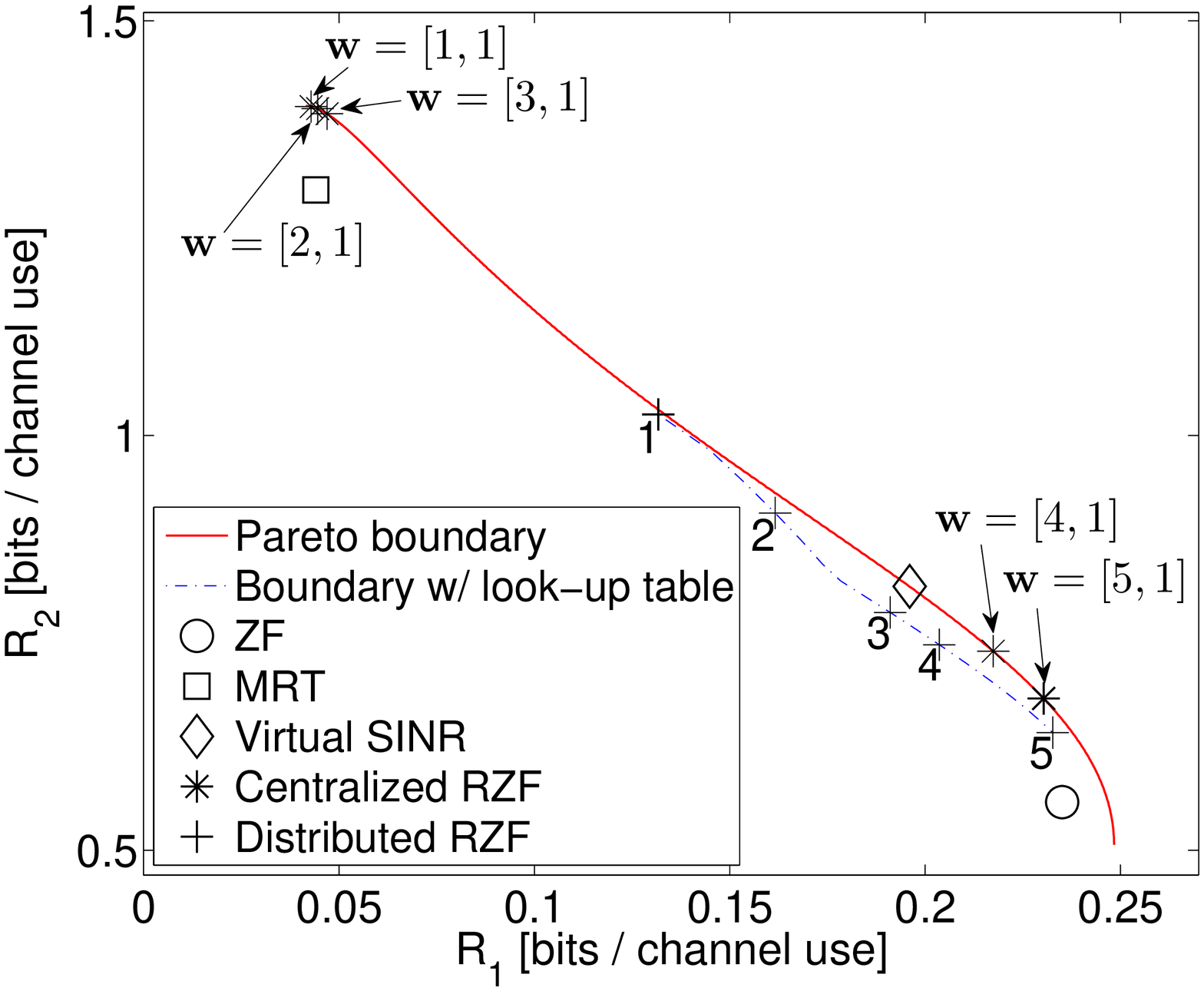}
} }
\vspace{2em}
\end{minipage}
\hspace{-2em}
\begin{minipage}{0.15\textwidth}
{\vspace{-2em}
\begin{tabular}{| c | c | c | }
\hline
 & $\alpha_{21}$ & $\alpha_{12}$  \\ \hline
1 & $ 0 $ & $0.8511$  \\ \hline
2 & $0.0667$ & $0.5780$  \\ \hline
3 & $0.2444$ & $0.3268$  \\ \hline
4 & $1.0889$ & $0.2393$  \\ \hline
5 & $2.2$ & $0.0735$  \\ \hline
\end{tabular}
}
\end{minipage}
\caption{Performance of RZFCB with the proposed rate control
algorithms: The centralized rate control, marked with $*$,
searches for the weighted sum rate maximizing point. (The weight
vector $\wbf$ is shown in the figure.) The distributed rate
control scheme, marked with $+$, sets the interference leakage
levels as shown in the table. In Figs. \ref{fig:Pareto1} (a) and
 (b), 'virtual SINR' denotes the rate-tuple obtained by the virtual SINR
 (or SLNR)  beamforming method in \cite{Zakhour&Gesbert:09WSA}. }
\label{fig:Pareto1}
\end{figure*}

Several advantages in the RZFCB  are
summarized below.

\noindent $\bullet$ Real-time fully distributed operation is possible based on the proposed heuristic control approach. Transmitter $i$ only needs to know
$\{\hbf_{ji}: j=1,2,\cdots,K\}$ and $\{\sigma_j^2:
j=1,2,\cdots,K\}$.

\noindent $\bullet$ Once transmitter $i$ knows
$\{\alpha_{1i},\cdots,\alpha_{i-1,i},\alpha_{i+1,i},
\alpha_{Ki}\}$, there exists a very fast algorithm, the SOPC
algorithm, to design the RZFCB beam vector. Furthermore, in the
case of $K=3$, there is an approximate closed-form solution.

\noindent $\bullet$ Transmitter $i$ knows its SINR and achievable rate exactly,
and its achievable rate is given by
$R_i  = \log\bigg( 1+\frac{| \hbf_{ii}^H\vbf_i
|^2}{(1+\epsilon_i)\sigma_i^2} \bigg)$.
So, transmission based on this rate will be successful with high
probability.  This is true even when $\{\alpha_{ji}\}$ are designed
suboptimally, i.e., away from the Pareto boundary of the rate
region. Thus, the RZFCB scheme is robust.

\noindent $\bullet$ On the contrary to the ZF scheme, RZFCB does not require $N \ge K$.

%\vspace{-1em}
%%%%%%%%%%%%%%%%%%%%%%%%%%%%%%%%%%%%%%%%%%%%%%%%%%%%
\section{RZFCB for MIMO Interference Channels}
\label{sec:MIMO_IC}
%%%%%%%%%%%%%%%%%%%%%%%%%%%%%%%%%%%%%%%%%%%%%%%%%%%

In this section, we  consider the case that both transmitters and
receivers are equipped with multiple antennas i.e., MIMO
interference channels. In the MIMO case, we consider the weighted sum rate maximization under the RZF framework and then propose a
solution to the MIMO RZFCB based on the projected gradient method
\cite{goldstein64}. The rate control idea in the MISO case can be
applied to the MIMO case too.

%\vspace{-1em}
\subsection{Problem Formulation} \label{subsec:MIMOformulation}
%\vspace{-0.5em}

We assume that each receiver has $M$ receive antennas and each
transmitter has $N$ transmit antennas. In this case, the received
signal at receiver $i$ is given by
\begin{equation}
\ybf_i = \Hbf_{ii}\Vbf_i\sbf_i + \sum_{j\neq i}\Hbf_{ij}\Vbf_j\sbf_j + \nbf_i,
\end{equation}
where $\Hbf_{ij}$ is the $M \times N$ channel matrix from
transmitter $j$ to receiver $i$, $\Vbf_i$ is the $N \times d_i$
beamforming matrix,  $\sbf_i$ is the $d_i \times 1$ transmit
symbol vector at transmitter $i$ from a  Gaussian codebook with
$\sbf_i \sim {\mathcal{CN}}(0,\Ibf_{d_i})$, and
$\nbf_i\sim{\mathcal{CN}}(0,\sigma_i^2\Ibf)$ is the additive noise.
As in the MISO case, we have a transmit power constraint,
$\|\Vbf_j\|_F^2 \le P_j$, for transmitter $j$. The proposed RZF constraint
in the MIMO case is given by an inequality with the Frobenius norm
as
\begin{equation}\label{eq:MIMO_RZF}
 \|\Hbf_{ji} \Vbf_{i} \|_F^2 \le \alpha_{ji}\sigma_j^2,
\quad \forall i, \ j \neq i
\end{equation}
for some constant $\alpha_{ji} \ge 0$. As in the MISO case, the RZF
constraints reduce to ZF constraints when $\alpha_{ji}=0$ for all
$i, j\neq i$. With the MIMO RZF constraints, a cooperative beam
design problem that maximizes the weighted sum rate is formulated as follows:
\vspace{0.5em}
\begin{problem}[RZF  cooperative beamforming problem]
\label{prob:co_MIMO}
\begin{eqnarray}
&\underset{\{\Vbf_i\}}{\mbox{max}} & ~\sum_{i=1}^{K} w_i \log\bigg|
\Ibf_M+(\sigma_i^2\Ibf+\Bbf_i)^{-1}\Hbf_{ii}\Vbf_{i}\Vbf_{i}^H\Hbf_{ii}^H\bigg|,
\label{eq:MIMO_Problem1} \nonumber\\
&\mbox{subject to}&
 {\mathrm{(C. 1)}} \quad \| {\Hbf}_{ji} {\Vbf}_i \|_F^2 \le
\alpha_{ji}\sigma_j^2,
\quad \forall i, j \neq i,  \label{eq:MIMO_RZFprobConstr1} \nonumber \\
&               &
{\mathrm{(C. 2)}} \quad \|\Vbf_{i}\|_F^2\le P_i, \quad \forall i,
\label{eq:MIMO_RZFprobConstr2}
\end{eqnarray}
where $w_i \ge 0$, $\sum_i w_i =1$, and  $\Bbf_i = \sum_{j\neq i}
\Hbf_{ij}\Vbf_{j}\Vbf_{j}^H\Hbf_{ij}^H$ is the interference
covariance matrix at receiver $i$.
\end{problem}

\noindent Note that, in Problem \ref{prob:co_MIMO}, the interference
from other transmitters is incorporated in the rate formula through
the interference covariance matrix $\Bbf_i$ capturing the residual
inter-cell interference under the RZF constraints. As in the MISO
case, we will derive a lower bound on the  rate of each user by exploiting the
RZF constraints to convert the joint design problem into a set of separate design problems. Note that, under the RZF constraints, the total
power of interference from  undesired transmitters is upper bounded
as
\begin{equation}\label{eq:MIMO_interf_power}
\mbox{tr}(\Bbf_i)
= \sum_{j\neq i} \|\Hbf_{ij}\Vbf_{j}\|_F^2
\le \sigma_i^2\sum_{j\neq i} \alpha_{ij}
=: \epsilon_i\sigma_i^2,
\end{equation}
which implies $\Bbf_i \le \epsilon_i\sigma_i^2 \Ibf$.

Hassibi and Hochwald derived a lower bound on the ergodic rate of
a MIMO channel with interference \cite{Hassibi&Hochwald:03IT}.
However, their result is not directly applicable
 here since the rate here is for an instantaneous channel
realization. Thus, we present a new lower bound under the RZF
interference constraints in the following Lemma.

\begin{lemma} \label{lemma:MIMO_LB}
A lower bound on the rate of receiver $i$ under the RZF constraints
is given by
\begin{eqnarray}\label{eq:MIMO_lower_sum_rate}
& &\hspace{-3em} \log\bigg| \Ibf_M+(\sigma_i^2\Ibf+\Bbf_i)^{-1}
\Hbf_{ii}\Vbf_i\Vbf_i^H\Hbf_{ii}^H \bigg| \nonumber \\
& &\hspace{1em} \ge
\log\bigg|\Ibf_M+\frac{1}{\sigma_i^2(1+\epsilon_i)}
\Hbf_{ii}\Vbf_i\Vbf_i^H\Hbf_{ii}^H \bigg|,
\end{eqnarray}
where $\mbox{tr}(\Bbf_i) = \sum_{j\neq i}\|\Hbf_{ij}\Vbf_j\|_F^2 \le \epsilon_i\sigma_i^2$ for all $i$.
\end{lemma}

\begin{proof}
The rate at receiver $i$ is given by
\begin{equation}
\log|\Ibf+\Phibf_i^{-1}\Abf_i| = \log\prod_{k=1}^M
(1+\lambda_k(\Phibf_i^{-1}\Abf_i))
\end{equation}
where $\Phibf_i=\sigma_i^2\Ibf+\Bbf_i$,
$\Abf_i=\Hbf_{ii}\Vbf_i\Vbf_i^H\Hbf_{ii}^H$, and $\lambda_k(\Xbf)$
denotes the $k$-th largest eigenvalue of $\Xbf$.  By the
Rayleigh-Ritz theorem \cite[p.176]{Horn&Johnson:book}, we have
\begin{equation}
\lambda_M(\Phibf_i^{-1}\Abf_i) \le
\frac{\xbf^H\Abf_i\xbf}{\xbf^H\Phibf_i\xbf} =
\frac{\pbf^H\Phibf_i^{-\frac{1}{2}}\Abf_i\Phibf_i^{-\frac{H}{2}}\pbf}{\pbf^H\pbf}
\le \lambda_1(\Phibf_i^{-1}\Abf_i)
\end{equation}
for any non-zero vector $\xbf\in{\mathbb{C}}^M$ and
$\pbf:=\Phibf_i^{\frac{H}{2}}\xbf$. From the Courant-Fischer
theorem \cite[p.179]{Horn&Johnson:book}, the $k$-th largest
generalized eigenvalue of $\Phibf_i^{-1}\Abf_i$, $k=1,\cdots, M$
is given by
\begin{eqnarray}
\lambda_k(\Phibf_i^{-1}\Abf_i) ~~~~ =
\underset{\substack{\pbf\neq 0,~ \pbf\in{\mathbb{C}}^M,\\
\pbf\bot\pbf_1,\cdots,\pbf_{k-1}}}{\max}
\frac{\pbf^H\Phibf_i^{-\frac{1}{2}}\Abf_i\Phibf_i^{-\frac{H}{2}}\pbf}{\pbf^H\pbf}
\end{eqnarray}
where $\pbf_i$ is the eigenvector associated with the $i$-th
largest eigenvalue of
$\Phibf_i^{-\frac{1}{2}}\Abf_i\Phibf_i^{-\frac{H}{2}}$. Let
$\Abf_i=\Ubf_i\Sigmabf_i\Ubf_i^H$ be the eigen-decomposition of
$\Abf_i$, where $\Sigmabf_i =
\mbox{diag}(\lambda_1(\Abf_i),\cdots,\lambda_M(\Abf_i))$. Then,
for all $k$
{\small
\begin{eqnarray} \lambda_k(\Phibf_i^{-1}\Abf_i) &=&
~~~~
\underset{\substack{\pbf\neq 0,~ \pbf\in{\mathbb{C}}^M,\\
\pbf\bot\pbf_1,\cdots,\pbf_{k-1}}}{\max}
\frac{\pbf^H\Phibf_i^{-\frac{1}{2}}\Abf_i\Phibf_i^{-\frac{H}{2}}\pbf}{\pbf^H\pbf}, \nonumber \\
&=& ~~~~ \underset{\substack{\pbf\neq 0,~ \pbf\in{\mathbb{C}}^M,\\
\pbf\bot\pbf_1,\cdots,\pbf_{k-1}}}{\max} \frac{
\pbf^H\Phibf_i^{-\frac{1}{2}}\Ubf_i
       \Sigmabf_i
       \Ubf_i^H\Phibf_i^{-\frac{H}{2}}\pbf }
     {\pbf^H\pbf}, \nonumber \\
&=& \underset{\substack{\zbf\neq 0,~ \zbf\in{\mathbb{C}}^M,\\
{\tiny \Phibf_i^{H/2}}\Ubf_i\zbf\bot \pbf_1,\cdots,\pbf_{k-1}
}}{\max}
\frac{\zbf^H\Sigmabf_i\zbf}{\zbf^H\Ubf_i^H\Phibf_i\Ubf_i\zbf}, \nonumber \\
& & \qquad\quad (\zbf:= \Ubf_i^H\Phibf_i^{-\frac{H}{2}}\pbf), \nonumber \\
&\stackrel{(a)}{\ge}&
\underset{\substack{\zbf \ne 0,~ ||\zbf||=1,~ \zbf\in{\mathbb{C}}^M,\\
z_{k+1} = z_{k+2} = \cdots = z_M = 0, \\
{\tiny \Phibf_i^{H/2}}\Ubf_i\zbf\bot
\pbf_1,\cdots,\pbf_{k-1},}}{\max}
\frac{\zbf^H\Sigmabf_i\zbf}{\zbf^H\Ubf_i^H\Phibf_i\Ubf_i\zbf}, \nonumber \\
& & \qquad\quad (\zbf=[z_1,z_2,\cdots,z_M]^T)\nonumber \\
&=&
\underset{\substack{\zbf \ne 0,~ ||\zbf||=1,~ \zbf\in{\mathbb{C}}^M,\\
z_{k+1} = z_{k+2} = \cdots = z_M = 0, \\
{\tiny \Phibf_i^{H/2}}\Ubf_i\zbf\bot
\pbf_1,\cdots,\pbf_{k-1},}}{\max}
\frac{\sum_{j=1}^k\lambda_j(\Abf_i)|z_j|^2}
{\zbf^H\Ubf_i^H\Phibf_i\Ubf_i\zbf }, \nonumber \\
&\stackrel{(b)}{\ge}&
\frac{\lambda_k(\Abf_i)}{\lambda_1(\Phibf_i)},
\label{eq:MIMOeigenBound}
\end{eqnarray}
}
where (a) is satisfied since the feasible set for $\zbf$ is reduced and (b) is satisfied
since $||\zbf||^2=|z_1|^2+\cdots+|z_k|^2=1$, $\lambda_1(\Abf_i)\ge \cdots \ge
\lambda_k(\Abf_i)$, and $\zbf^H\Ubf_i^H\Phibf_i\Ubf_i\zbf \le
\lambda_1(\Phibf_i)$ by Rayleigh-Ritz theorem.
Based on \eqref{eq:MIMOeigenBound}, a lower bound on the
rate is given by
\begin{align}
\log|\Ibf+\Phibf_i^{-1}\Abf_i|
\ge
\left(1+\frac{\lambda_k(\Abf_i)}{\lambda_1(\Phibf_i)}\right).
\end{align}
Since $\Phibf_i = \sigma_i^2\Ibf + \sum_{j\neq i} \Hbf_{ij}\Vbf_j\Vbf_j^H\Hbf_{ij}^H$, we have
$\lambda_1(\Phibf_i) = \sigma_i^2 + \lambda_1\left(\sum_{j\neq i}
\Hbf_{ij}\Vbf_j\Vbf_j^H\Hbf_{ij}^H\right)$,  where the maximum
eigenvalue of the interference covariance matrix is upper bounded
by
$\lambda_1\Big(\sum_{j\neq i}
\Hbf_{ij}\Vbf_j\Vbf_j^H\Hbf_{ij}^H\Big)
 \le \mbox{tr}\Big(\sum_{j\neq i} \Hbf_{ij}\Vbf_j\Vbf_j^H\Hbf_{ij}^H\Big)
= \sum_{j\neq i}\alpha_{ij}\sigma_i^2=\epsilon_i\sigma_i^2$.
Thus,  a lower bound of rate at receiver $i$ is given by
\[
|\Ibf+\Phibf_i^{-1}\Abf_i|
\ge
\Big|\Ibf_M+\frac{1}{\sigma_i^2(1+\epsilon_{i})}\Abf_i\Big|.
\]
\end{proof}

\noindent Note that in \eqref{eq:MIMO_lower_sum_rate} the
inter-user dependency is removed and the beam design can be
performed at each transmitter in a distributed manner. Based on
the lower bound \eqref{eq:MIMO_lower_sum_rate}, the RZFCB problem
is now formulated as a distributed problem:
\begin{problem}[The MIMO RZFCB problem] \label{prob:MIMORZFCBProblem}
\begin{eqnarray} \label{eq:MIMOProblem}
&\underset{\Vbf_i}{\mbox{max}}& \phi_i(\Vbf_i) :=
\log\left|\Ibf_M+\frac{1}{\sigma_i^2(1+\epsilon_i)}\Hbf_{ii}\Vbf_{i}\Vbf_{i}^H\Hbf_{ii}^H
\right|, \nonumber \\
&\mbox{subject to}& {\mathrm{(C.1)}} \quad \|\Hbf_{ji} \Vbf_{i}\|_F^2 \le
\alpha_{ji}\sigma_j^2, \quad \forall j\neq i,      \nonumber \\
&               & {\mathrm{(C.2)}} \quad \|\Vbf_{i}\|_F^2\le P_i,
\end{eqnarray}
for each transmitter $i = 1, 2, \cdots, K$.
\end{problem}

\noindent Note that Problem \ref{prob:MIMORZFCBProblem} is now fully distributed.
One of several known algorithms for constrained optimization
can be used to solve Problem \ref{prob:MIMORZFCBProblem} for given
$\{\alpha_{ji}\}$. In particular, we choose to use the projected
gradient method (PGM) by Goldstein \cite{goldstein64}. The proposed PGM-based beam design algorithm for MIMO ICs is provided in
Table \ref{table:algorithm}. Detailed explanation of the beam design with PGM algorithm is
provided in \cite{Park&Lee&Sung&Yukawa:12Arxiv}.

\begin{figure*}[t]
\centerline{
\SetLabels
\L(0.16*-0.1) (a) \\
\L(0.49*-0.1) (b) \\
\L(0.82*-0.1) (c) \\
\endSetLabels
\leavevmode
%\ShowGrid
\strut\AffixLabels{
\scalefig{0.33}\epsfbox{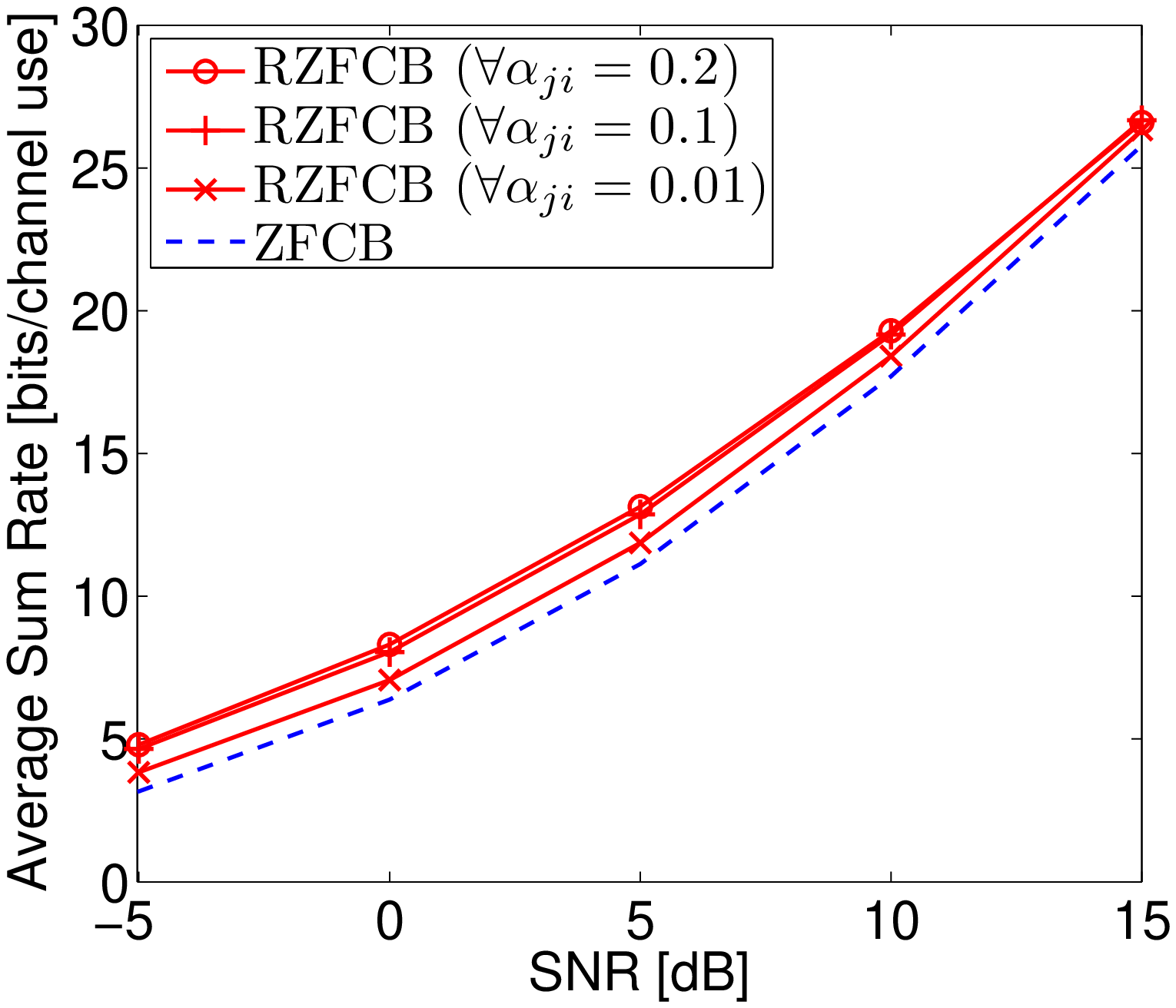}
\scalefig{0.33}\epsfbox{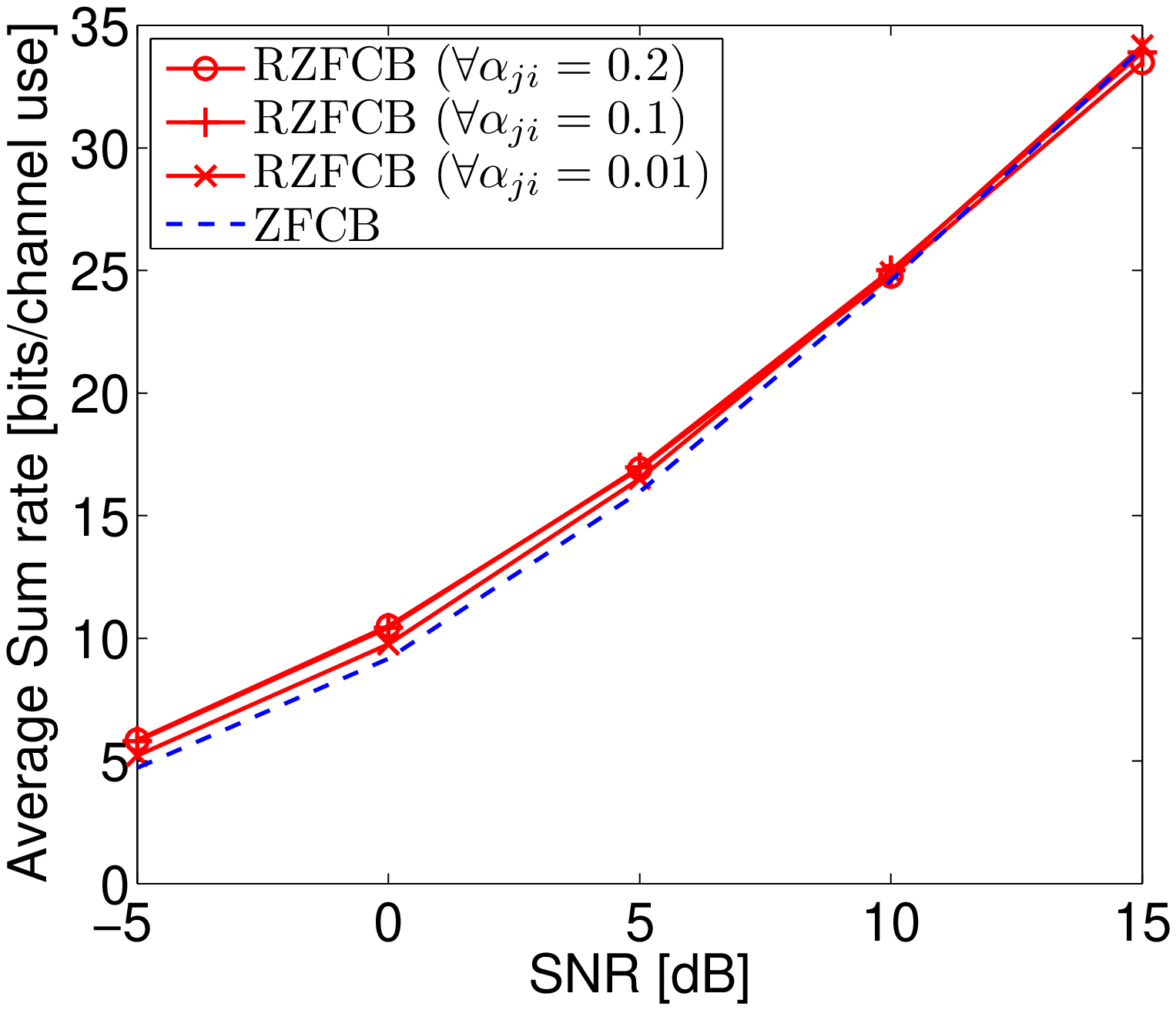}
\scalefig{0.33}\epsfbox{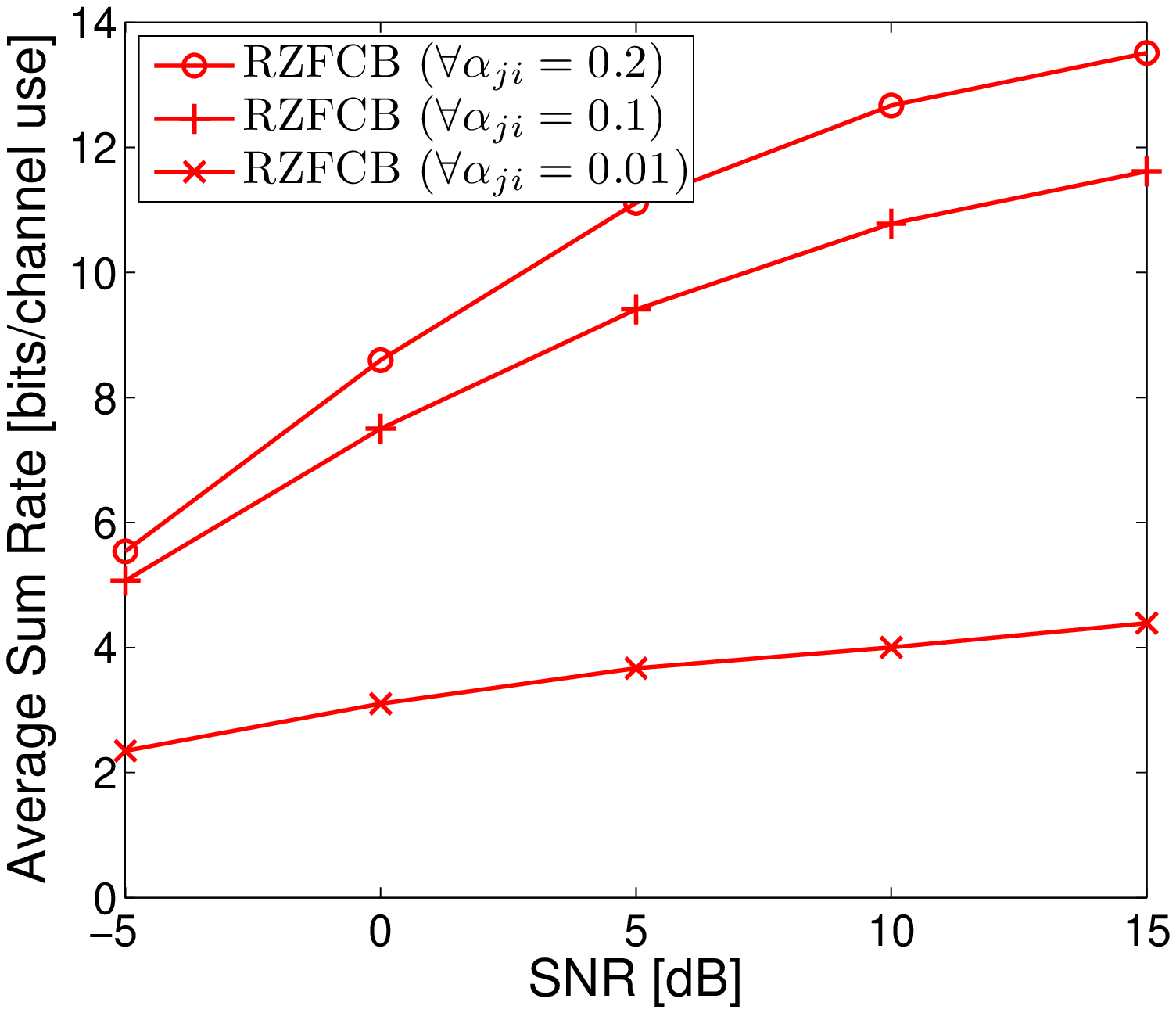} } }
\vspace{1.5em} \caption{Sum rate of RZFCB: (a) (K,M,N)=(3, 2, 6),
(b) (K,M,N)=(3, 2, 8), and (c) (K,M,N)=(4, 2, 6).} \label{fig:sumrate}
\end{figure*}

%%%%%%%%%%%%%%%%%%%%%%%%%%%%%%%%%%%%%%%%%%%%% Algorithm %
\begin{table}[ht]
\caption{Beam design algorithm for MIMO IC using PGM.}
\label{table:algorithm} \centering
\begin{tabular}{p{310pt}}
\hline \vspace{-0.8em}
\normalsize{  For each transmitter $i\in\{1,\cdots, K\}$,

\textit{0.} Initialize $\Vbf_i$ as the ZF beamforming matrix.

\textit{1.} Compute gradient of $\phi(\Vbf_i)$.

\textit{2.} Perform a steepest descent shift of $\Vbf_i$.

\textit{3.} Perform successive metric projections of $\Vbf_i$ onto
 constraint sets.

\textit{4.} Go to Step 1 and repeat until the relative difference
of $\phi_i(\Vbf_i)$ is

~~ less than a pre-determined threshold. }
\\
\hline
\end{tabular}
\end{table}

\vspace{-2em}

\subsection{Numerical Results}

In this section, we provide some numerical results for the
performance of RZFCB in the MIMO case.
We consider three MIMO
interference channels  with system parameters $(K,M,N)=(3, 2, 6)$,
$(3, 2, 8)$, and $(4, 2, 6)$. In each case, we set $\alpha_{ji}=0.01, 0.1$ and $0.2$ for all $i$ and $j$. The step size parameter
for the PGM is chosen to be $0.01$ for all iterations.
Figures \ref{fig:sumrate} (a), (b), and (c) show the sum rate
performance of the ZFCB and RZFCB averaged over 30 independent
channel realizations. In Fig. \ref{fig:sumrate} (a) it is seen that  the RZFCB
outperforms the ZFCB at all SNR and the gain of the RZFCB over the
ZFCB at low SNR is large when $N=KM$. This large gain at low SNR is especially
important because most cell-edge receivers operate in the low SNR
regime. In Fig. \ref{fig:sumrate} (b) it is seen that the ZF scheme performs well when the number of TX antenna is more than enough and the dimension of ZF beams is large, as expected.  In the case of $N < KM$ as in Fig. (c), the ZFCB is infeasible but the RZFCB still works well.

%\vspace{-0.5em}
%%%%%%%%%%%%%%%%%%%%%%%%%%%%%%%%%%%%%%%%%%%%%%%%%%%%%%%%%%%%%%%%%%%%
\section{Conclusion}
\label{sec:conclusion}
%%%%%%%%%%%%%%%%%%%%%%%%%%%%%%%%%%%%%%%%%%%%%%%%%%%%%%%%%%%%%%%%%%%%
%\vspace{-0.5em}

We have considered coordinated beamforming for MISO  and MIMO
interference channels under the RZF framework. In the MISO case, we
have shown that the SOPC strategy with a set of well chosen
interference relaxation levels is Pareto-optimal. We have provided (approximate) closed-form solutions for the
SOPC strategy in the cases of two and three users and the SOPC
algorithm in the general case for a given set of
interference relaxation levels. In the MIMO case, we have
formulated the RZFCB problem  as a distributed optimization
problem  based on a newly derived rate lower bound and have
provided an algorithm based on the PGM to solve the MIMO RZFCB
beam design problem. We have also considered the rate control
problem under the RZFCB framework and have provided a centralized
approach and a fully-distributed heuristic approach to control the
rate-tuple location roughly along the Pareto boundary of the achievable
rate region. Numerical results validate the RZFCB paradigm.

\vspace{-0.5em}
%%%%%%%%%%%%%%%%%%%%%%%%%%%%%%%%%%%%%%%%%%%%%%%%%%%

%%%%%%%%%%%%%%%%%%%%%%%%%%%%%%%%%%%%%%%%%%%%%%%%%%%

\end{document}

%%%%%%%%%%%%%%%%%%%%%%%%%%%%%%%%%%%%%%%%%%%%%%%%%%%